\pdfoutput=1
\documentclass[11pt]{article}
\usepackage{fullpage}
\usepackage{mathtools,amsfonts,amsthm,amssymb,bbm}
\usepackage{xcolor}
\usepackage[backref,colorlinks,citecolor=blue,linkcolor=magenta,bookmarks=true]{hyperref}
\usepackage[nameinlink]{cleveref}
\usepackage[algo2e,ruled,linesnumbered]{algorithm2e}

\newcommand{\1}[1]{\mathbbm{1}\left[#1\right]} 

\newcommand{\argmin}{\operatorname*{argmin}}
\newcommand{\bbF}{\mathbb{F}}
\newcommand{\Binomial}{\mathsf{Binomial}} 

\newcommand{\Egood}{\mathcal{E}^{\mathsf{good}}}
\newcommand{\eps}{\epsilon}
\newcommand{\Ex}[2]{\operatorname*{\mathbb{E}}_{#1}\left[#2\right]} 
\newcommand{\F}{\mathcal{F}}
\newcommand{\N}{\mathcal{N}} 
\newcommand{\poly}{\mathrm{poly}}
\newcommand{\polylog}{\mathrm{polylog}}
\newcommand{\pr}[2]{\Pr_{#1}\left[#2\right]} 
\newcommand{\R}{\mathbb{R}} 
\newcommand{\rank}{\mathrm{rank}}
\newcommand{\rmd}{\mathrm{d}}
\newcommand{\sgn}{\mathrm{sgn}} 
\newcommand{\Span}{\mathrm{Span}}

\newtheorem{problem}{Problem}
\newtheorem{lemma}{Lemma}
\newtheorem{theorem}[lemma]{Theorem}
\newtheorem{definition}[lemma]{Definition}
\newtheorem{corollary}[lemma]{Corollary}

\newtheorem{conjecture}{Conjecture}
\newtheorem{remark}[lemma]{Remark}

\title{A Combinatorial Approach to Robust PCA}
\author{
Weihao Kong\thanks{Google Research.  Email: \texttt{weihaokong@google.com}.}
\and
Mingda Qiao\thanks{University of California, Berkeley. Email: \texttt{mingda.qiao@berkeley.edu}. Part of this work was done while the author was an intern at Google Research and a graduate student at Stanford University.}
\and
Rajat Sen\thanks{Google Research. Email: \texttt{senrajat@google.com}.}
}
\date{}

\begin{document}

\maketitle

\begin{abstract}
    We study the problem of recovering Gaussian data under adversarial corruptions when the noises are low-rank and the corruptions are on the coordinate level.
    Concretely, we assume that the Gaussian noises lie in an unknown $k$-dimensional subspace $U \subseteq \mathbb{R}^d$, and $s$ randomly chosen coordinates of each data point fall into the control of an adversary. This setting models the scenario of learning from high-dimensional yet structured data that are transmitted through a highly-noisy channel, so that the data points are unlikely to be entirely clean.

    Our main result is an efficient algorithm that, when $ks^2 = O(d)$, recovers \emph{every single} data point up to a nearly-optimal $\ell_1$ error of $\tilde O(ks/d)$ in expectation. At the core of our proof is a new analysis of the well-known Basis Pursuit (BP) method for recovering a sparse signal, which is known to succeed under additional assumptions (e.g., incoherence or the restricted isometry property) on the underlying subspace $U$. In contrast, we present a novel approach via studying a natural combinatorial problem and show that, over the randomness in the support of the sparse signal, a high-probability error bound is possible even if the subspace $U$ is \emph{arbitrary}.
\end{abstract}

\section{Introduction}
The flurry of recent work on algorithmic robust statistics studies various settings of learning and estimation in the presence of adversarial outliers (see e.g.,~\cite{DK19} for a survey). Techniques have been developed for designing algorithms that are provably correct and both sample- and computationally-efficient. In most of these settings, however, an important assumption is that the majority of the data points are clean, so that non-trivial estimation is possible at least information-theoretically. Alternative models such as list-decodable learning (\cite{BBV08,CSV17}), while allowing a majority of the data to be corrupted, still require the clean data to constitute a non-negligible portion of the dataset.

This work explores a setting of robust learning from high-dimensional and structured data that are densely corrupted. Concretely, for some $d \gg k$, each clean data point $x$ is a $d$-dimensional vector lying in a $k$-dimensional subspace $U \subseteq \R^d$. When $d$ is huge, each data point is no longer accurately modeled by an atomic object that is either entirely clean or fully under the control of the adversary. For instance, when the data are transmitted through a noisy channel, it is conceivable that only a small fraction of the data \emph{entries} get corrupted. Even if each entry gets corrupted independently with a small probability $\eps$, the fraction of entirely clean data points, $(1 - \eps)^d$, could still be exponentially small. Indeed, we will focus on the extreme case that \emph{every} data point may contain $s$ corrupted entries, for some $s \ll d$.

Unfortunately, even in a seemingly benign setup where both $k$ and $s$ are small, an all-powerful adversary might still prevent us from extracting any information from the data. For instance, when the subspace $U$ consists of the vectors supported on the first $k$ coordinates, an adversary that corrupts $s \ge k$ entries per sample can easily wipe out all the information. However, not all hope is lost---in this example, the critical reason for the adversary to succeed is that the corruptions are perfectly aligned with the informative entries in the data. In many practical scenarios (e.g., when the corruptions result from the noisy transmission channel), it is reasonable to assume that the \emph{locations} of corruptions are independent of the structure of the subspace $U$, so the worst-case scenario above is unlikely to happen. Concretely, we will assume in this work that the corruption on each data point is restricted to a randomly chosen subset of $s$ coordinates.

The model that we informally defined above is closely related to the vast body of literature on \emph{sparse recovery} and \emph{robust PCA}. When the clean data point lies in a \emph{known} subspace in our setup, the problem is equivalent to recovering the (additive) corruption, which is $s$-sparse, given a linear measurement of it; see Remark~\ref{remark:basis-pursuit} for details. This is exactly the problem of sparse recovery, though we note that most previous work focused on the structures of the linear measurements that enable accurate and efficient recovery, whereas we consider an alternative setup where the linear measurement is essentially arbitrary, but the support of the sparse vector is uniformly distributed. If we arrange the observed data points into a matrix, the problem can be rephrased as decomposing this matrix into the sum of a low-rank matrix (consisting of the clean data) and a sparse matrix (consisting of the additive corruptions). This is known as robust PCA, and prior work has identified various suites of assumptions---on both the incoherence of the low-rank component and the support of the sparse component---under which an accurate decomposition can be efficiently computed. We discuss the related work on these two problems in more detail in Section~\ref{sec:related}. 

\subsection{Problem Setup}

\paragraph{Notations.} $[n]$ is the set $\{1, 2, \ldots, n\}$. We denote by $\binom{S}{s}$ the family of all size-$s$ subsets of a set $S$. For vector $v$, $\|v\|_p$ is its $p$-norm when $p \ge 1$, and $\|v\|_0$ is the number of non-zero entries in $v$. For $v \in \R^d$ and $S \subseteq [d]$, $v|_S \in \R^{|S|}$ is the restriction of $v$ to the coordinates with indices in $S$. Similarly, for a matrix $A$ with $d$ columns, $A|_S$ is the restriction of $A$ to the columns indexed by $S \subseteq [d]$. The $\tilde O(\cdot)$ and $\tilde\Omega(\cdot)$ notations hide $\polylog(d)$ factors, where $d$ is the ambient dimension in the problem.

We first formally state the problem of recovering a vector lying in a \emph{known} low-dimensional subspace when several randomly chosen entries are corrupted. This will be the key step in the recovery of an entire dataset that lie in an unknown subspace, a problem that we define later.

\begin{problem}[Recovery of One Data Point]\label{prob:sparse}
We are given a $k$-dimensional subspace $U \subseteq \R^d$ in which the unknown $x$ lies. Set $S$ is chosen uniformly at random from $\binom{[d]}{s}$. The adversary produces the corrupted data point $\tilde x$ subject to:
\begin{itemize}
    \item Corruptions are restricted to $S$: $\tilde x_i = x_i$ for every $i \in [d] \setminus S$.
    \item Corruptions are bounded: $\|\tilde x - x\|_{\infty} \le B$.
\end{itemize}
Our goal is to output an estimate $\hat x$ of the clean data $x$ with a small error.
\end{problem}

Note that we may assume $B = 1$; otherwise we simply scale everything down by a factor of $B$. Also note that another natural corruption model---in which each coordinate is independently included in $S$ with probability $\eps$---corresponds to the case that $s \sim \Binomial(d, \eps)$ in Problem~\ref{prob:sparse}. Later, we will analyze an algorithm for Problem~\ref{prob:sparse} that is oblivious to the value of $s$, so the guarantee for such an algorithm immediately implies a guarantee under the independent corruption setup.

The boundedness assumption $\|\tilde x - x\|_{\infty} \le B$ in Problem~\ref{prob:sparse} arises naturally in the context of the following problem of recovering a Gaussian dataset. In this setup, we can easily obtain a crude estimate for each coordinate of the mean vector, and then safely disregard the data entries that are far from the estimated mean.

\begin{problem}[Recovery of Structured Data under Coordinate-Level Corruption]\label{prob:mean-est}
    The clean data $x^{(1)}$, $x^{(2)}$, $\ldots$, $x^{(n)}$ are independent samples from an unknown Gaussian distribution $\N(\mu, \Sigma)$, where $\rank(\Sigma) = k$. Sets $S^{(1)}, \ldots, S^{(n)}$ are chosen independently and uniformly at random from $\binom{[d]}{s}$. The adversary produces the corrupted data $\tilde x^{(1)}, \ldots, \tilde x^{(n)}$ such that $\tilde x^{(i)}_j = x^{(i)}_j$ for every $j \in [d]\setminus S^{(i)}$. Our goal is to recover the clean data given the corrupted data.
\end{problem}

Note that a solution to Problem~\ref{prob:mean-est} immediately allows us to post-process the recovered data and extract information about the underlying distribution (e.g., the mean $\mu$).

\subsection{Our Results}
We analyze the natural $\ell_1$ convex relaxation, termed the Basis Pursuit (BP) method by Chen, Donoho and Saunders~\cite{CDS01}, for Problem~\ref{prob:sparse}: Given the corrupted data point $\tilde x$ and subspace $U$, we solve the following optimization problem to obtain an estimate of $x$:
\begin{equation}\begin{split}\label{eq:basis-pursuit}
    \text{minimize} \quad&\|\hat x - \tilde x\|_1\\
    \text{subject to}\quad& \hat x \in U
\end{split}\end{equation}

Note that the boundedness parameter $B$ and the sparsity $s$ of the corruption are not used in~\eqref{eq:basis-pursuit}.

\begin{remark}\label{remark:basis-pursuit}
Basis pursuit is more often formulated as minimizing $\|x\|_1$ subject to $Ax = y$, where $x$ is a sparse signal, $A$ specifies a linear measurement, and $y$ is the observation from the measurement. The optimization above is equivalent to this more usual formulation: we may write $\delta = \tilde x - \hat x$ and let $A$ be the $(d - k) \times d$ matrix with row space $U^{\bot}$, the orthogonal complement of $U$. The constraint $\hat x \in U$ in \eqref{eq:basis-pursuit} is equivalent to $A(\tilde x - \delta) = 0$. Thus, Program~\eqref{eq:basis-pursuit} is equivalent to the minimization of $\|\delta\|_1$ subject to $A\delta = A\tilde x$.
\end{remark}

Our main result states that the BP program in~\eqref{eq:basis-pursuit}, while being a simple and natural relaxation of the sparsity objective, approximately denoises the data with high probability, even if no assumptions are made on the subspace $U$.

\begin{theorem}\label{thm:sparse-recovery}
    In the setup of Problem~\ref{prob:sparse} with $B = 1$, let $x^*$ be an optimum of the BP method. For any $t > 0$, over the randomness in the set of corrupted coordinates, $\pr{}{\|x^* - x\|_1 \ge t}$ is upper bounded by both
    \[
         \frac{1}{(\lfloor t/4\rfloor + 1)!}\cdot\left(\frac{12s^2k}{d}\right)^{\lfloor t/4\rfloor + 1}
    \]
    and
    \[
        \frac{24ks}{d}.
    \]
    Furthermore, assuming $s \le \sqrt{d / 2}$, the probability is also upper bounded by
    \[
         12k \cdot \left(\frac{2s}{d}\right)^{1 + \lfloor t/(48k)\rfloor}.
    \]
\end{theorem}

The second bound above holds uniformly for all $t > 0$. The third bound is sharper for larger values of $t$: assuming $s = O(\sqrt{d})$, the probability of incurring error $\ge t$ is at most $O(k)\cdot(2s/d)^{\Omega(t/k)}$, which decays exponentially as $t$ grows. The first bound, in comparison, removes the $1/k$ factor in the exponent at the cost of a slightly larger base in the exponentiation.

An easy calculation (deferred to Appendix~\ref{sec:expcted-error}) gives an upper bound on the expected $\ell_1$ error.
\begin{corollary}[Expected error]\label{cor:expected-error}
    In the setup of Problem~\ref{prob:sparse}, the expected $\ell_1$ error of the BP method is $\tilde O\left(\frac{Bks}{d}\cdot\max\left\{\frac{ks^2}{d}, 1\right\}\right)$.
\end{corollary}

Assuming $ks^2 = O(d)$, the error bound in Corollary~\ref{cor:expected-error} reduces to $\tilde O(Bks/d)$, which cannot be improved in general---the expected $\ell_1$ error is at least $\Omega(Bks/d)$ when the subspace is $U = \Span(\{e_1, e_2, \ldots, e_k\})$. This is because, in expectation, $ks/d$ out of the first $k$ coordinates get corrupted, and there is no hope in recovering each corrupted entry with an error better than $\Omega(B)$.

In comparison, the naïve algorithm that simply outputs $\tilde x$ gives an error of $\Omega(Bs)$ in the worst case, which is much worse in the low dimensional ($k \ll d$) regime. Another straightforward algorithm is based on the observation that, since $U$ is $k$-dimensional, there exist $k$ coordinates that uniquely determine a vector in $U$. With probability $\ge 1 - \frac{ks}{d}$, none of them gets corrupted and $x$ can be perfectly recovered by solving a linear system. However, when at least one of the $k$ coordinates does get corrupted, it is unclear that the recovery satisfies a reasonable error bound.

The setup of Problem~\ref{prob:mean-est} differs from the sparse recovery problem in two aspects: (1) We don't know the subspace that contains the clean data points; (2) The corruptions are no longer bounded. Our next theorem shows that, in the $ks = O(d)$ regime, the subspace can be efficiently recovered.

\begin{theorem}\label{thm:subspace-recovery}
    In the setup of Problem~\ref{prob:mean-est}, let $U$ be the column space of $\Sigma$. Assuming $ks \le c_0 \cdot d$ for sufficiently small constant $c_0 > 0$, there is an algorithm (namely, Algorithm~\ref{algo:subspace-recovery}) that recovers the subspace $U' \coloneqq \Span(U \cup \{\mu\})$ with high probability using $n = \tilde O(k^2)$ samples and in $\poly(d)$ time.
\end{theorem}

The theorem above implicitly assumes that we can access the samples as well as perform arithmetic operations with unlimited accuracy. To translate this result to more realistic models of computation (with limited precision), we note that the bottleneck in Algorithm~\ref{algo:subspace-recovery} will be solving an instance of robust linear regression in which the marginal distribution is a Gaussian with the covariance matrix being a  full-rank principal submatrix of $\Sigma$. Thus, the algorithm will likely be numerically stable as long as every invertible principal submatrix of $\Sigma$ is reasonably well-conditioned (compared to the bit precision). See Remark~\ref{remark:numerical-stability-detailed} for a more detailed discussion.

Theorems \ref{thm:sparse-recovery}~and~\ref{thm:subspace-recovery} together give the following guarantee for recovering an entire Gaussian dataset.

\begin{theorem}\label{thm:mean-estimation}
    In the setup of Problem~\ref{prob:mean-est}, assuming $n = \tilde \Omega(k^2)$ and $ks \le c_0 \cdot d$ for some sufficiently small constant $c_0 > 0$, there is an algorithm (namely, Algorithm~\ref{algo:mean-est}) that outputs $n$ estimates $\hat x^{(1)}, \ldots, \hat x^{(n)}$ in $\poly(d, n)$ time, such that each $\|\hat x^{(i)} - x^{(i)}\|_1$ is bounded by $\tilde O\left(\frac{Bks}{d}\cdot\max\left\{\frac{ks^2}{d}, 1\right\}\right)$ in expectation, where $B \coloneqq \sqrt{\max_{i \in [d]}\Sigma_{ii}}$.
\end{theorem}

As a direct corollary, the recovered data points allow us to estimate the unknown mean $\mu$ of the underlying Gaussian distribution.

\begin{corollary}\label{cor:mean-estimation}
    Under the same assumptions as in Theorem~\ref{thm:mean-estimation}, there is an efficient algorithm that outputs $\hat \mu$ such that $\|\hat\mu - \mu\|_1$ is $\tilde O\left(\frac{Bks}{d}\cdot\max\left\{\frac{ks^2}{d}, 1\right\} + \frac{Bd}{\sqrt{n}}\right)$ in expectation.
\end{corollary}

The additional error term $Bd/\sqrt{n}$ is unavoidable, even when $k = 1$ and no corruption is present. When $ks^2 = O(d)$ and $n = \Omega(d^4/k^2s^2)$, the error bound in Corollary~\ref{cor:mean-estimation} reduces to $\tilde O(Bks/d)$. Note that the na\"ive approach of estimating each coordinate of the mean vector separately would give a worse error guarantee when $k \ll d$. Indeed, when restricted to the $i$-th coordinate, the dataset becomes samples drawn from $\N(\mu_i, \Sigma_{ii})$, in which an $O(s/d)$ fraction gets corrupted. We can at best estimate $\mu_i$ up to an additive error of $O(\sqrt{\Sigma_{ii}}\cdot s/d) = O(Bs/d)$, so the $\ell_1$ error bound would be $O(Bs)$, i.e., larger by a factor of $d / k$. In light of this comparison, Corollary~\ref{cor:mean-estimation} formalizes the intuition that since the problem is essentially $k$-dimensional, we only need to pay the $O(Bs/d)$ cost (the error in the one-dimensional case) $k$ times instead of $d$ times.

\section{An Overview of Our Approach}\label{sec:approach}

We sketch the proofs of Theorems \ref{thm:sparse-recovery}, \ref{thm:subspace-recovery}~and~\ref{thm:mean-estimation} here. Our analysis of the Basis Pursuit method for Problem~\ref{prob:sparse} is inspired by the Restricted Nullspace Property (RNP), which is known to imply the exact recovery guarantee of BP. We introduce a quantitative version of RNP (Lemma~\ref{lemma:nec-cond}) that gives a necessary condition---on both the set $S$ of corrupted coordinates and the subspace $U$---for BP to incur a large error in its recovery.

The technical highlight of our analysis is Lemma~\ref{lemma:comb-constraint}, which gives a combinatorial constraint on the family of size-$s$ subsets of $[d]$ on which the aforementioned condition holds. Crucially, this constraint is a consequence of only the dimensionality of $U$, and none of the usual structural assumptions (such as incoherence or restricted isometry) on $U$. We then sketch how this constraint implies the desired high-probability error bound via studying a combinatorial problem similar to those in the extremal combinatorics literature.

Then, we explain how Problem~\ref{prob:mean-est} is solved as an application of the above. To recover the unknown subspace in the $ks = O(d)$ regime (Theorem~\ref{thm:subspace-recovery}), the key observation is that the recovery problem is essentially $(k + 1)$-dimensional, so the known results on robust linear regression (against sample-level corruptions) can be applied. Theorem~\ref{thm:mean-estimation} then follows from the two earlier theorems and a simple reduction of the problem to the bounded-corruption case.

\subsection{A New Analysis of Basis Pursuit}
\paragraph{Generalization of the RNP.} With respect to subspace $U \subseteq \R^d$, we say that $S \subseteq [d]$ is \emph{$t$-dominant} for $t > 0$, if there exists $u \in U$ such that $\|u\|_1 = 1$ and
\[
    \sum_{i \in S}|u_i| \cdot \1{|u_i| \le 1/t} \ge 1/2.
\]
We will show in Lemma~\ref{lemma:nec-cond} that the BP method incurs an $\ell_1$ error of $t$ only if the set $S$ of corrupted coordinates is $t$-dominant.

The definition above becomes more natural in the context of the Restricted Nullspace Property, on which the textbook recovery guarantee of BP (e.g.,~\cite{FR13}) relies.\footnote{As mentioned in Remark~\ref{remark:basis-pursuit}, subspace $U$ in Problem~\ref{prob:sparse} corresponds to the nullspace of the design matrix $A$ in the usual formulation of BP, so Definition~\ref{definition:RNP} is indeed about a nullspace.}

\begin{definition}[Restricted Nullspace Property]\label{definition:RNP}
    A linear subspace $U \subseteq \R^d$ satisfies $s$-RNP if it holds for every $u \in U \setminus \{0\}$ and $S \in \binom{[d]}{s}$ that $\|u|_S\|_1 < \|u|_{[d]\setminus S}\|_1$.
\end{definition}

$s$-RNP is easily seen to be equivalent to the non-existence of $1$-dominant sets of size $s$---when $t = 1$, the condition $\1{|u_i| \le 1/t}$ in the definition of $t$-dominant sets is implied by $\|u\|_1 = 1$. 

It is well-known that RNP is equivalent to the exact recovery guarantee of the Basis Pursuit method.

\begin{theorem}[Theorem 7.8 of~\cite{Wainwright19}]
    In the setup of Problem~\ref{prob:sparse}, $U$ satisfies $s$-RNP if and only if BP recovers every $x \in U$ with zero error for any $S \in \binom{[d]}{s}$.
\end{theorem}

In particular, the BP method may err when the RNP is violated. For example, the subspace $U = \Span(\{e_1, e_2, \ldots, e_k\})$ does not satisfy $s$-RNP for any $s \ge 1$, since $[s]$ is a $1$-dominant set witnessed by $u = e_1$. As discussed earlier, exact recovery in this setup is not guaranteed, since whenever $S \cap [k]$ is non-empty, the adversary may arbitrarily corrupt one of the $k$ informative coordinates.

On the bright side, when $S$ is uniformly drawn from $\binom{[d]}{s}$, BP still satisfies a high-probability error bound on this example. We can easily verify that for this choice of $U$, the BP method simply zeros out the last $d-k$ coordinates of $\tilde x$, and the $\ell_1$ error is linear in $|S \cap [k]|$, which is $ks/d$ in expectation and has a sub-exponential tail.

This high-probability error bound can be alternatively explained in terms of of $t$-dominant sets. It can be easily verified that for $t \le k$, each fixed $S \in \binom{[d]}{s}$ is $t$-dominant w.r.t.\ $U$ if and only if $|S \cap [k]| \ge t/2$. When $S$ is chosen uniformly at random, because of the sub-exponential tail of $|S \cap [k]|$, the probability of $|S \cap [k]| \ge t/2$ is exponentially small in $t$. Lemma~\ref{lemma:nec-cond} would then imply a high-probability error bound of BP. Perhaps surprisingly, our Theorem~\ref{thm:sparse-recovery} gives quantitatively similar bounds for every $k$-dimensional subspace $U$. In other words, the simple example above is essentially the worst case.

\paragraph{A constraint on dominant sets.} The crux of our analysis is then to upper bound the number of $t$-dominant sets w.r.t.\ the subspace $U$. As a warmup, we start with the special case that $t = 1$, and assume that the dominance is always \emph{strict}, in the sense that every $1$-dominant set $S \in \binom{[d]}{s}$ has a witness $u \in U$ such that $\|u\|_1 = 1$ and $\|u|_S\|_1 > 1/2$ (instead of $\|u|_S\|_1 \ge 1/2$).

Let $\F$ be the family of strictly $1$-dominant sets. We note the following property of $\F$:
\begin{center}
\it
The family $\F$ does not contain $k + 1$ sets that are pairwise disjoint.
\end{center}
To see this, suppose towards a contradiction that $S_1, S_2, \ldots, S_{k+1} \in \F$ are pairwise disjoint, and vector $u^{(i)} \in U$ witnesses the strict dominance of $S_i$. Consider the $(k + 1)\times d$ matrix $A = \begin{bmatrix}
    u^{(1)} & u^{(2)} & \cdots & u^{(k+1)} 
\end{bmatrix}^{\top}$. We obtain another $(k + 1)\times(k + 1)$ matrix $A'$ by appropriately aggregating the columns of $A$. First, if $j \in S_i$ and $A_{ij} < 0$, we negate the $j$-th column of $A$. Then, we define the $i$-th column of $A'$ as the sum of the $j$-th column of $A$ over $j \in S_i$. For instance, for $(S_1, S_2, S_3) = (\{1, 2\}, \{3, 4\}, \{5\})$ and
\[
    A = \begin{bmatrix}
        2 & -1 & 0 & 0 & 2\\
        3 & 3 & -4 & 5 & 2\\
        1 & -3 & 0 & 2 & -7
    \end{bmatrix},
\]
Columns 2, 3, and 5 of $A$ are negated, and the resulting matrix $A'$ is
\[
    A' = \begin{bmatrix}
        3 & 0 & -2\\
        0 & 9 & -2\\
        4 & 2 & 7        
    \end{bmatrix}.
\]

It is clear from our construction of $A'$ that: (1) $A'$ is at most rank-$k$, since $\rank(A') \le \rank(A) \le \dim(U) = k$; (2) $A'$ is strictly diagonally dominant, i.e., $|A'_{ii}| > \sum_{j \ne i}|A'_{ij}|$ for every $i \in [k + 1]$. However, these two are contradictory, since every strictly diagonally dominant matrix is full-rank by the Gershgorin circle theorem.

\paragraph{Bounding the number of dominant sets.} Then, bounding the number of $1$-dominant sets becomes a combinatorial problem: What is the largest possible size of a family $\F \subseteq \binom{[d]}{s}$ that contains no $k + 1$ pairwise disjoint members?  This problem has been studied since the work of Erd\H{o}s~\cite{Erdos65}, and the long-standing Erd\H{o}s Matching Conjecture (EMC) states that when $d \ge s(k+1) - 1$, the maximum size is exactly
\[
    \max\left\{\binom{d}{s} - \binom{d - k}{s}, \binom{s(k+1) - 1}{s}\right\}.
\]
Note that the above is the maximum between $|\F_1|$ and $|\F_2|$, where $\F_1 \coloneqq \{S \in \binom{[d]}{s}: S \cap [k] \ne \emptyset\}$ and $\F_2 \coloneqq \binom{[s(k+1)-1]}{s}$ are both families that contain no $k + 1$ disjoint sets via simple pigeonhole arguments. A recent work of Frankl and Kupavskii~\cite{FK22} proved the EMC for all sufficiently large $k$ and $d \ge \frac{5}{3}ks$.

While these proofs tend to be technical and sophisticated, if we are willing to lose a constant factor, a nearly-tight bound of $|\F| \le O(ks/d)\cdot\binom{d}{s}$ is well-known and easy to prove via the following argument. Let $A_1, A_2, \ldots, A_{\lfloor d/s\rfloor}$ be pairwise disjoint members of $\binom{[d]}{s}$ chosen uniformly at random.\footnote{More formally, we choose a permutation $\pi: [d]\to [d]$ uniformly at random, and let $A_i \coloneqq \{\pi((i-1)s+1), \pi((i-1)s+2),\ldots,\pi(is)\}$.} What would be the expected size of $\F \cap \{A_1, A_2, \ldots, A_{\lfloor d/s\rfloor}\}$? This number must be at most $k$; otherwise we would find $k + 1$ disjoint sets in $\F$. Meanwhile, this expectation is exactly given by $|\F|\cdot\lfloor d/s\rfloor / \binom{d}{s}$, since each $A_i$ coincides with each $S \in \F$ with probability $1/\binom{d}{s}$. The two observations together give $|\F| \le k\binom{d}{s}/\lfloor d/s\rfloor = O(ks/d) \cdot \binom{d}{s}$. Furthermore, this bound is tight up to a constant factor, as the set family $\F_1$ defined earlier has size $\binom{d}{s} - \binom{d-k}{s} \ge \Omega(\min\{ks/d, 1\})\cdot\binom{d}{s}$. 


\paragraph{Handling weak dominance and the general $t$ case.} Recall that we made two simplifications: strict dominance  and $t = 1$. Without the strict dominance, repeating the first half of the above argument would end up with a $(k+1)\times(k+1)$ matrix $A'$ that is only weakly diagonally dominant, i.e., $|A'_{ii}| \ge \sum_{j\ne i}|A'_{ij}|$, which does not imply that $A'$ is full rank.\footnote{Consider the $2\times 2$ all-one matrix.} Fortunately, we still have $\rank(A') = \Omega(k)$.\footnote{Technically, this also requires $A'$ to have non-zero rows in addition to weak diagonal dominance.} Thus, if we started with $Ck$ independent $1$-dominant sets for sufficiently large $C$, we would still obtain a contradiction to $\dim U = k$. The lower bound on $\rank(A')$ is proved by a probabilistic argument: We carefully sample a submatrix of $A'$, so that the submatrix is strictly diagonally dominant, and the dimension only shrinks by a constant factor in expectation.

To handle the general $t$ case, we strengthen the combinatorial constraint on $\F$ into the following condition (Lemma~\ref{lemma:comb-constraint}, our main technical lemma):
\begin{center}
    \it
    The family of $t$-dominant sets does not contain $\Omega(k)$ sets that are ``almost disjoint''.
\end{center}
We will use two different notions of ``almost disjointness'': (1) the pairwise intersections are bounded by $O(t/k)$; (2) the sets become disjoint after removing $O(t)$ elements from each set. Note that the first definition implies the second, so, conversely, using the second notion imposes a stricter constraint on the family. Also, using either notion, the constraint gets stricter as $t$ grows.

Lemma~\ref{lemma:comb-constraint} is proved using an analogous (yet slightly more technical) argument to the above---assuming the existence of $\Omega(k)$ almost disjoint $t$-dominant sets, we construct a matrix using the witnesses of the dominance, such that the matrix has rank $\le \dim U = k$ and satisfies a weaker diagonal dominance property. Finally, these two are shown to be contradictory via a probabilistic argument.

Using Lemma~\ref{lemma:comb-constraint}, we upper bound the number of $t$-dominant sets via two different approaches. The third bound in Theorem~\ref{thm:sparse-recovery} follows from the first notion of almost disjointness, and an easy extension of the double counting argument presented above. The first bound in Theorem~\ref{thm:sparse-recovery}, on the other hand, relies on the finer-grained definition of disjointness and a quite different analysis based on Hall's marriage theorem.

\subsection{Recovery of Subspace and the Entire Dataset}
We briefly discuss how we recover the subspace in the setup of Problem~\ref{prob:mean-est}, thereby reducing the problem to recovering the individual data points given the low-rank structure (i.e., Problem~\ref{prob:sparse}).

Let $U$ be the column space of the Gaussian covariance $\Sigma$. For simplicity, we first assume $\mu = 0$ and that the first $k$ coordinates uniquely determine a vector in $U$, i.e., there exist $c_{k+1}, c_{k+2}, \ldots, c_d \in \R^k$ such that for every $u \in U$ and $i \in \{k + 1, k + 2, \ldots, d\}$, $u_i = c_i^{\top}u|_{[k]}$. Then, the problem can be reduced to $d - k$ linear regression instances in $k$ dimensions: Let $\Sigma'$ be the $k \times k$ restriction of $\Sigma$ to coordinates in $[k]$. For each clean data $x^{(j)} \in \R^d$, $\left(x^{(j)}|_{[k]}, x^{(j)}_i\right)$ can be viewed an instance drawn from $\N(0, \Sigma')$ and its label produced by the linear function $x \mapsto c_i^{\top}x$. Consequently, $\left(\tilde x^{(j)}|_{[k]}, \tilde x^{(j)}_i\right)$ is a potentially corrupted copy and the probability of corruption is at most $(k + 1)\cdot s/d$, which is bounded by a small constant assuming $ks = O(d)$. Using the robust linear regression algorithm of~\cite{DKS19}, we can efficiently recover each $c_i$, and thus the whole subspace $U$.

In the general case, we can easily find a subset of $k$ pivotal coordinates, since the algorithm of~\cite{DKS19} also allows us to test whether including an additional coordinate introduces a linear dependence. Also note that the above procedure only recovers the column space $U$ of the covariance, which might not contain the mean vector $\mu$. To include $\mu$ into the subspace, it is sufficient to compute $\mu^{\top}v_1, \mu^{\top}v_2, \ldots, \mu^{\top}v_{d-k}$, where $\{v_i\}_{i \in [d - k]}$ is a basis of $U^{\bot}$. Note that for $x \sim \N(\mu, \Sigma)$, $x^{\top}v_i$ exactly gives $\mu^{\top}v_i$, though it might be altered if any of the coordinates in the support of $v_i$ is corrupted. If we pick each $v_i$ to be $(k + 1)$-sparse (which is always possible), the fraction of incorrect values of $x^{\top}v_i$ is at most $(k + 1)\cdot (s/d) \ll 1 / 2$. This allows us to recover $\mu^{\top}v_i$ by a simple majority vote.

Finally, we sketch the proof of Theorem~\ref{thm:mean-estimation}: First, we find the subspace $U' = \Span(U \cup \{\mu\})$ as outlined above. Then, we estimate each $\mu_i$ up to an $O(Bs/d)$ additive error, by solving the one-dimensional mean estimation problem with corruption rate $s / d$ and variance $\Sigma_{ii} \le B^2$. We then clip the $i$-th coordinate of the data points so that they become $\tilde O(B)$-close to the estimated mean; with high probability, no clean entries will be changed. Therefore, each clipped data point can be viewed as an instance of Problem~\ref{prob:sparse} with known subspace $U'$ and corruptions bounded by $\tilde O(B)$. Each data point can then be recovered up to a small $\ell_1$ error using the BP method.

\section{Related Work}\label{sec:related}
\paragraph{Comparison to [Liu-Park-Rekatsinas-Tzamos, ICML'21].} Most closely related to our work is a recent paper of Liu, Park, Rekatsinas and Tzamos~\cite{LPRT21} on robust Gaussian mean estimation in a very similar setup. In their model, the clean data lie in the column space of a $d \times k$ matrix $A$ (where $k \ll d$), and the adversary is allowed to corrupt a small fraction of coordinates in the dataset. An important difference is that we assume that the corruptions are randomly located, while the locations might be adversarially chosen in the model of~\cite{LPRT21}. In return, we make no assumptions on the subspace, whereas their estimation guarantee depends on a complexity measure of matrix $A$. Therefore, the results in the two models are qualitatively different and incomparable.
 
In more detail, the positive results of~\cite{LPRT21} rely on the assumption that the linear code defined by $A$ has a large distance, which guarantees that the clean data can be decoded (information-theoretically) even with a small number of missing or corrupted coordinates. The mean estimation bounds in~\cite{LPRT21} depend on the quantity $m_A$ defined as the smallest number $m$ such that we can remove $m$ rows from $A$ to reduce its rank by $1$. It is easy to show that this definition of $m_A$ is equivalent to $\min_{z \in \R^k\setminus\{0\}}\|Az\|_0$, namely, the distance of the linear code defined by $A$.

In terms of the algorithm, we follow a similar approach to that of~\cite{LPRT21}: first denoising each data point using the low-rank structure, and then use the cleaned data to estimate the mean. While the assumptions of~\cite{LPRT21} makes the denoising information-theoretically easy, the decoding is NP-hard in general. In contrast, we analyze the natural $\ell_1$ relaxation of the sparse recovery problem, which, via a non-trivial analysis, is shown to succeed with high probability over the randomness in the locations of the corrupted coordinates.

\paragraph{Other models of coordinate-level corruption.} Chen, Caramanis and Mannor~\cite{CCM13} considered a sparse linear regression setup, in which the weight vector has at most $k \ll d$ non-zero entries, and on each coordinate, the adversary may corrupt an $\eps$ fraction of the data entries (that are not necessarily randomly located). This corruption model is stronger than ours, since an $\eps$ fraction of the data points could be entirely corrupted. Their main result is an efficient algorithm that outputs a non-trivial estimate of the weight vector whenever the sample size is $\Omega(k\log d)$ and the corruption rate is $O\left(\frac{1}{\sqrt{k}\log d}\right)$. 

Hu and Reingold~\cite{HR21} studied a robust mean estimation setting in which the data can be highly incomplete: the adversary is allowed to corrupt an $\eps$ fraction of the samples as well as to conceal all but a $\gamma$ fraction of the entries on each coordinate. The concealed entries are chosen independently of the data. They gave efficient algorithms that recover the mean of the data distribution up to a dimension-independent $\ell_2$ error of $\tilde O(\eps/\gamma)$ or $O(\sqrt{\eps/\gamma})$, depending on whether the distribution is assumed to be sub-Gaussian or just bounded-covariance.

\paragraph{Robust PCA.} Robust PCA is the problem of decomposing a matrix into a low-rank component and a sparse component. In the seminal work of Cand\`es, Li, Ma, and Wright~\cite{CLMW11}, it was shown that under certain assumptions, such a decomposition is unique and can be computed efficiently using a simple convex program. Roughly speaking, the unique decomposition is possible when the low-rank component is far from being sparse (guaranteed by an appropriate incoherence assumption), and the sparse component is far from being low-rank (guaranteed by the assumption that the corruptions are randomly located). Subsequent work~\cite{HKZ11,ANW12,XCS12,NNSAJ14} studied alternative formulations of the setup and optimization methods.

In the context of robust PCA, we may view the corrupted data in Problem~\ref{prob:mean-est} as an $n \times d$ matrix. Our main result (Theorem~\ref{thm:mean-estimation}) guarantees the accurate recovery of the data matrix up to a small $\ell_1$ error on each row. Compared to most of the work on robust PCA, our algorithm does not need any assumption on the low-rank subspace, at the cost of stronger distributional assumptions (i.e., the data are Gaussian), a more narrow parameter regime in which our results apply (i.e., the assumption that $ks = O(d)$), and an imperfect recovery guarantee (which is inevitable).

We remark that \cite{ANW12} is a notable exception: their assumptions on the low-rank component is milder than incoherence. When their result is specialized to robust PCA on a $d_1 \times d_2$ matrix, it requires each entry of the low-rank component to be bounded by $\alpha/\sqrt{d_1d_2}$ in absolute value, where $d_1$ and $d_2$ are the dimensions of the matrix. In this scenario, their result implies the recovery of the low-rank part up to a squared Frobenius error of $O(\alpha^2s'/(d_1d_2))$, where $s'$ is the number of corrupted entries. In the setup of our Problem~\ref{prob:mean-est}, we have $s' = s\cdot n$ and $\alpha/\sqrt{d_1d_2} = O(B\sqrt{\log(nd)})$ for $B = \sqrt{\max_{i \in [d]}\Sigma_{ii}}$.\footnote{We may first estimate each coordinate up to an $O(Bs/d)$ error, and subtract this estimate from the data. After the shifting, the clean data are bounded by $O(B\sqrt{\log(nd)})$ entrywise.} Then, the result of~\cite{ANW12} implies an $\tilde O(B^2sn)$ squared Frobenius error, which is no better than using the corrupted data as the estimate. On the other hand, as noted in \cite[Example 6]{ANW12}, this bound cannot be further improved in their setup, since the locations of the sparse corruptions are allowed to be arbitrary.

\paragraph{Outlier-robust PCA.} Another line of work on outlier-robust PCA~\cite{FXY12,XCM13,YX15,JLT20,KSKO20,DKPP23} studies a different and orthogonal model of robustness. In this setup, the goal is to identify the top eigenvector of an unknown distribution, given a corrupted dataset from that distribution. Compared to the robust PCA literature in the previous paragraph, this setting is further from our setting, since the corruptions are on the sample level, and the goal is to extract information of the underlying distribution, rather than recovering the dataset.

\paragraph{Sparse recovery.} In sparse recovery, we aim to recover a $k$-sparse vector $x \in \R^d$ (also called the signal) given the value of $y = Ax$ for an $n \times d$ matrix $A$ (also called the measurements). The Basis Pursuit program, minimizing $\|\hat x\|_1$ subject to $A\hat x = y$, is a natural convex relaxation for the problem. Unlike the setup of our Problem~\ref{prob:sparse}, the linear measurement $A$ is often under the control of the recovery algorithm, and the goal is to design $A$ such that efficient recovery is possible; in contrast, the linear measurements (induced by the unknown subspace $U$) are arbitrary in Problem~\ref{prob:sparse}. The sparse recovery problem is also closely related to compressed sensing~\cite{CRT06,Donoho06} as well as data stream algorithms. A different formulation has been studied in the TCS literature: even when $x$ is not $k$-sparse, we are still required to output a $k$-sparse $\hat x$ such that $\|\hat x - x\|_p$ is close to $\min_{\|x'\|_0 \le k}\|x' - x\|_p$. Nearly tight performance guarantees have been obtained for various different metrics of performance guarantee~\cite{IR08,BIPW10,PW11,IP11}, and also in a related setup where the linear measurements can be chosen adaptively~\cite{IPW11,PW13,KP19}.

\paragraph{Robust subspace recovery.} Robust subspace recovery is the problem of identifying a subspace that contains the \emph{inliers} (i.e., clean data) in the presence of outliers that might be outside this subspace. \cite{HM13} gave an algorithm that is provably correct and efficient as long as: (1) The fraction of inliers is at least $k/d$, where $k$ and $d$ are the dimensions of the subspace and the ambient space, respectively; (2) The inliers and outliers together satisfy a non-degeneracy assumption. \cite{ML19} studies the case where only the inliers are assumed to be in general position and the outliers can be arbitrary. Recent works of~\cite{RY20,BK21} gave algorithms in the list-decodable learning setting based on the sum-of-squares method, under additional distributional assumptions on the inliers. See \cite{LM18} for a detailed survey of this literature.

\paragraph{Extremal combinatorics.} As outlined in Section~\ref{sec:approach}, our analysis of the BP method involves upper bounding the size of a set family subject to certain combinatorial constraints (e.g., on the pairwise intersections between the sets in the family). This is the focus of extremal combinatorics. In particular, our proof of Lemma~\ref{lemma:size-upper-bound} is based on techniques in this literature. For many problems in extremal combinatorics, the goal is to either find the exact maximum, or at least obtain asymptotically tight bounds (i.e., tight up to a $1 + o(1)$ factor). In some cases, the set families that achieve this maximum size have also been characterized. In contrast, for the purpose of this work, we are more interested in bounds that might be off by a constant factor, but hold for a wider range of parameters.

Most closely related to our approach is the Erd\H{o}s Matching Conjecture regarding families of size-$s$ subsets of $[d]$ that do not contain $k + 1$ pairwise disjoint sets. \cite{Erdos65} conjectured that the maximum size is given by the larger between $\binom{d}{s} - \binom{d - k}{s}$ and $\binom{s(k+1)-1}{s}$, for all $d \ge s(k+1) - 1$. It was also shown in~\cite{Erdos65} that the first term gives the correct answer for fixed $k, s$ and sufficiently large $d \ge d_0(k, s)$.  This result was strengthened by subsequent work, and Frankl and Kupavskii~\cite{FK22} recently relaxed the assumption to $d \ge \frac{5}{3}ks$ for all sufficiently large $k$; we refer the readers to~\cite{FK22} for a more detailed discussion on the progress towards proving EMC. Another notable special case is the $k = 1$ case, in which the family is not allowed to contain any two disjoint sets. A more general version of this case, now known as the Erd\H{o}s-Ko-Rado Theorem, was proved in~\cite{EKR61}; see \cite{Katona72} for a simplified proof and~\cite{DF83} for a survey of several generalizations and analogues of this theorem.

Part of our proof (Lemma~\ref{lemma:packing-number}) also requires constructing a large \emph{packing}, i.e., a subset of $\binom{[d]}{s}$ such that the size of any pairwise intersection among the family is at most $\Delta - 1$.\footnote{This is also called a ``combinatorial design'' in the TCS literature.} It is easy to see that the size of any packing is at most $\frac{\binom{d}{\Delta}}{\binom{s}{\Delta}}$: every member in the packing contains $\binom{s}{\Delta}$ subsets of size $\Delta$, and no two members may share the same size-$\Delta$ subset.
A Theorem of R\"odl~\cite{Rodl85} shows that this bound is asymptotically tight: for fixed $s, \Delta$ and as $d \to +\infty$, the largest packing has size $(1 - o(1))\cdot\frac{\binom{d}{\Delta}}{\binom{s}{\Delta}}$. Unfortunately, this theorem cannot be directly applied, since it does not hold uniformly for all values of $(d, s, \Delta)$ (even up to a constant factor). Instead, we use a well-known construction based on polynomials (which, e.g., was used by Nisan~\cite{Nisan91} for constructing pseudorandom generators) that nearly achieves this limit in the $s = O(\sqrt{d})$ regime.

\section{Proof of the Sparse Recovery Guarantee}
In this section, we prove Theorem~\ref{thm:sparse-recovery}, which states that every optimum $x^*$ of the Basis Pursuit (BP) program below is close to the clean data $x$ with high probability.
\begin{align*}
    \min_{\hat x} \quad&\|\hat x - \tilde x\|_1\\
    \text{subject to}\quad& \hat x \in U
\end{align*}

\subsection{A Necessary Condition for Large Error}
We start by stating a simple necessary condition for BP to incur a large error. Recall that in Problem~\ref{prob:sparse}, $U$ is the subspace in which $x$ lies and $S \subseteq [d]$ is the set of corrupted coordinates.

\begin{lemma}\label{lemma:nec-cond}
    In the setup of Problem~\ref{prob:sparse} with $B = 1$, if BP incurs an $\ell_1$ error of $t > 0$, $S$ is $t$-dominant with respect to $U$ in the sense that there exists $u \in U$ such that $\|u\|_1 = 1$ and
    \[
        \sum_{i \in S}|u_i|\cdot\1{|u_i| \le \frac{1}{t}} \ge \frac{1}{2}.
    \]
\end{lemma}

Note that when $t \le 1$, the inequality in Lemma~\ref{lemma:nec-cond} reduces to $\|u|_S\|_1 \ge 1/2$, which is exactly the negation of the restricted nullspace property (Definition~\ref{definition:RNP}).

\begin{proof}
    We first note that it is sufficient to analyze the $x = 0$ case, since both the problem setting and the BP method are invariant under a translation within $U$. Formally, let $\tilde x$ be the corrupted version of $x$, and $x^*$ be an optimum of BP (Program~\eqref{eq:basis-pursuit}). Then, for the alternative clean data $x' = 0$, $(\tilde x)' \coloneqq \tilde x - x$ is a valid corrupted version of $x'$ (in the setup of Problem~\ref{prob:sparse}), and $(x^*)' \coloneqq x^* - x$ would be an optimum of BP. The error would still be $\|(x^*)' - x'\|_1 = \|x^* - x\|_1 = t$.

    Let $\tilde x$ be a corrupted copy of the clean data $x = 0$. Let $x^* \in \argmin_{\hat x \in U}\|\hat x - \tilde x\|_1$ be an optimum with $\|x^*\|_1 = t$ and $u \coloneqq \frac{x^*}{\|x^*\|_1}$. Clearly, $u \in U$ and $\|u\|_1 = 1$. Furthermore, since $tu = x^* \in \argmin_{\hat x \in U}\|\hat x - \tilde x\|_1$, we have $tu \in \argmin_{\hat x \in \Span(\{u\})}\|\hat x - \tilde x\|_1$, or equivalently,
    \[
        t \in \argmin_{\alpha \in \R}\|\alpha u - \tilde x\|_1.
    \]
    The objective function above can be expanded as
    \[
        \|\alpha u - \tilde x\|_1
    =   \sum_{i=1}^{d}|u_i\alpha - \tilde x_i|
    =   \sum_{i \in [d]: u_i \ne 0}|u_i|\cdot\left|\alpha - \frac{\tilde x_i}{u_i}\right| + \sum_{i \in [d]: u_i = 0}\left|\tilde x_i\right|.
    \]
    This is equivalent to finding the weighted median of $d$ points on a line, where the $i$-th point is at location $\frac{\tilde x_i}{u_i}$ and has a weight of $|u_i|$. Therefore, in order for $t$ be a minimum, the weight in the region $[t, +\infty)$ must be at least half of the total weight (which is $\|u\|_1 = 1$), i.e.,
    \[
        \sum_{i=1}^d|u_i|\cdot\1{\frac{\tilde x_i}{u_i} \ge t} \ge \frac{1}{2}.
    \]
    Recall that in Problem~\ref{prob:sparse}, $\tilde x$ is required to satisfy: (1) $|\tilde x_i| \le B = 1$; (2) $\tilde x_i = 0$ for $i \notin S$. Thus, the left-hand side above is maximized when we pick $\tilde x_i = 0$ for $i \notin S$ and $\tilde x_i = \sgn(u_i)$ for every $i \in S$. This gives the desired inequality $\sum_{i \in S}|u_i|\cdot\1{|u_i| \le \frac{1}{t}} \ge \frac{1}{2}$.
\end{proof}

\subsection{A Combinatorial Constraint on Dominant Subsets}
It remains to show that, for large $t$ and \emph{any} low-dimensional subspace $U$, over the randomness in the corruption set $S$, $S$ is unlikely to be $t$-dominant as defined in Lemma~\ref{lemma:nec-cond}. Formally, we fix the $k$-dimensional subspace $U \subseteq \R^d$, size $s \in [d]$ and the error parameter $t > 0$. Define $\F_{t}^{(s)}$ as the family of size-$s$ subsets of $[d]$ that are $t$-dominant, i.e.,
\[
    \F_t^{(s)} \coloneqq
\left\{S \in \binom{[d]}{s}: \exists u \in U\text{ such that }\|u\|_1 = 1, \sum_{i \in S}|u_i|\cdot\1{|u_i| \le 1/t} \ge \frac{1}{2}\right\}.
\]
To upper bound the size of $\F_t^{(s)}$, we give the following combinatorial property of this family.

\begin{lemma}\label{lemma:comb-constraint}
    Let $\alpha \in \left(0, \frac{1}{2}\right)$,
    $k' \coloneqq \left\lfloor\frac{4(1-\alpha)k}{\alpha}\right\rfloor + 1$, $t' \coloneqq \lfloor(1/2 - \alpha)t\rfloor$ and $\Delta \coloneqq \left\lfloor\frac{t'}{k'-1}\right\rfloor$. For any $k'$ sets $S_1, S_2, \ldots, S_{k'} \in \F_t^{(s)}$:
    \begin{itemize}
        \item There are no pairwise disjoint sets $T_1, T_2, \ldots, T_{k'}$ such that $T_i \subseteq S_i$ and $|S_i \setminus T_i| \le t'$ for every $i \in [k']$.
        \item In particular, there must be a pair $S_i$ and $S_j$ such that $|S_i \cap S_j| > \Delta$ and $i \ne j$.
    \end{itemize}
\end{lemma}

Concretely, by setting $\alpha = 1/4$, we get $k' = 12k + 1$, $t' = \lfloor t/4\rfloor$ and $\Delta = \lfloor t/(48k)\rfloor$. In general, for every fixed $\alpha$, the first part of Lemma~\ref{lemma:comb-constraint} says that we cannot find $\Omega(k)$ sets in $\F_t^{(s)}$ that become disjoint after removing $O(t)$ elements from each, while the ``in particular'' part says that $\F_t^{(s)}$ cannot contain $\Omega(k)$ sets such that every pairwise intersection is of size $O(t/k)$. 

\begin{proof}
    Suppose towards a contradiction that $S_1, \ldots, S_{k'}$ are $k'$ sets in $\F_t^{(s)}$ that violate either of the two conditions. We will derive a contradiction in the following three steps: First, we show that violating the second condition implies a violation of the first. Then, we use the fact that the sets are in $\F_t^{(s)}$, together with the violation of the first condition, to obtain a matrix with a diagonal dominance property. Such a matrix is constructed so that its rank is at most the dimension of $U$. Finally, we show that the diagonal dominance implies a lower bound on the rank of the matrix, which contradicts the assumption that $U$ is low-dimensional.
    
    \paragraph{Reduce to the first condition.} Suppose that the second condition is violated, i.e., there are sets $S_1, \ldots, S_{k'} \in \F_t^{(s)}$ with small ($\le \Delta$) pairwise intersections. Then, we define
    \[
        T_i
    \coloneqq S_i\setminus\left(\bigcup_{j\in[k']\setminus \{i\}}S_j\right).
    \]
    Clearly, $T_1, T_2, \ldots, T_{k'}$ are pairwise disjoint and each $T_i$ is a subset of $S_i$. Furthermore, $|S_i \setminus T_i|$ is at most $\sum_{j \in [k']\setminus \{i\}}|S_i \cap S_j| \le (k'-1)\Delta \le t'$. This shows that the first condition is also violated.
    
    \paragraph{Construct a diagonally dominant matrix.} For each $i \in [k']$, let $u^{(i)} \in U$ be a witness of $S_i \in \F_t^{(s)}$, i.e., $\left\|u^{(i)}\right\|_1 = 1$ and
    \begin{equation}\label{eq:dominance-of-S_i}
        \sum_{j\in S_i}|u^{(i)}_j|\cdot\1{|u^{(i)}_j| \le 1/t} \ge \frac{1}{2}.
    \end{equation}
    It follows that
    \begin{align*}
            \sum_{j\in T_i}|u^{(i)}_j|
    &\ge    \sum_{j\in T_i}|u^{(i)}_j|\cdot\1{|u^{(i)}_j| \le 1/t}\\
    &=      \sum_{j\in S_i}|u^{(i)}_j|\cdot\1{|u^{(i)}_j| \le 1/t} - \sum_{j\in S_i \setminus T_i}|u^{(i)}_j|\cdot\1{|u^{(i)}_j| \le 1/t} \tag{$T_i \subseteq S_i$}\\
    &\ge    \frac{1}{2} - |S_i \setminus T_i|\cdot\frac{1}{t} \tag{Equation~\eqref{eq:dominance-of-S_i}}\\
    &\ge    \frac{1}{2} - \left(\frac{1}{2} - \alpha\right) = \alpha. \tag{$|S_i \setminus T_i| \le (1/2 - \alpha)t$}
    \end{align*}
    In words, an $\alpha$-fraction of the $1$-norm of $u^{(i)}$ is contributed by the coordinates in $T_i$. 
    
    Now, we construct a $k'\times k'$ matrix $A$ using the vectors $u^{(1)}, \ldots, u^{(k')}$, and then argue that the rank of $A$ is at most $k$. Roughly speaking, each entry $A_{ij}$ is obtained by aggregating the entries $\left\{u^{(i)}_l: l \in T_j\right\}$ with some of the entries negated, so that the diagonal entries of $A$ are maximized. Formally, we pick $s \in \{\pm 1\}^d$ such that:
    \[
        s_j = \begin{cases}
            +1, & j \in T_i~\text{and}~u^{(i)}_j \ge 0~\text{for some}~i,\\
            -1, & j \in T_i~\text{and}~u^{(i)}_j < 0~\text{for some}~i,\\
            \text{arbitrary}, & j \notin \bigcup_{i\in[k']} T_i.
        \end{cases}
    \]
    Then, we let $A_{ij}
    \coloneqq \sum_{l \in T_j}s_ju^{(i)}_j$.
    
    It is easy to verify that $A$ is well-defined, and the following hold for every $i \in [k']$:
    \begin{itemize}
        \item $A_{ii} = \sum_{j \in T_i}\left|u^{(i)}_j\right| \ge \alpha$.
        \item $\sum_{j=1}^{k'}|A_{ij}| \le \sum_{j \in [k']}\sum_{l \in T_j}|u^{(i)}_l| \le \left\|u^{(i)}\right\|_1 = 1$.
    \end{itemize}
    Furthermore, $A$ can be obtained from the $k' \times d$ matrix
    \[
        \begin{bmatrix}
            u^{(1)} & u^{(2)} & \cdots & u^{(k')}
        \end{bmatrix}^{\top}
    \]
    via a sequence of elementary column operations (specifically, swapping two columns and adding a column or its negation to another column) and column deletions. Thus, the rank of $A$ is at most $\dim U = k$.
    
    \paragraph{Lower bound the rank via a probabilistic argument.} We scale the rows of $A$, so that the absolute values of the entries in each row sum up to $1$. Clearly, in the resulting matrix, we still have $A_{ii} \ge \alpha$ for each $i \in [k']$.
    
    We consider a randomly chosen set $J \subseteq [k']$, such that each element is included in $J$ with probability $p \coloneqq \frac{\alpha}{2\cdot(1-\alpha)}$ independently. We say that index $i \in [k']$ is \emph{dominant} if: (1) $i \in J$; (2) $\sum_{j \in J \setminus \{i\}}|A_{ij}| < \alpha$. By our choice of $J$, the two events are independent. The probability of the former is $p$, while the probability of the latter is, by Markov's inequality, at least
    \[
    1 - \frac{\Ex{J}{\sum_{j \in J \setminus \{i\}}|A_{ij}|}}{\alpha}
    =   1 - \frac{p \cdot \sum_{j \in [k']\setminus\{i\}}|A_{ij}|}{\alpha}
    \ge 1 - \frac{p\cdot(1 - \alpha)}{\alpha} = \frac{1}{2}.
    \]
    Therefore, each index $i$ is dominant with probability at least $p \cdot \frac{1}{2} = \frac{\alpha}{4\cdot(1 - \alpha)}$.
    
    By an averaging argument, there exists a deterministic set $J \subseteq [k']$ such that the number of dominant indices is at least $\left\lceil \frac{\alpha}{4\cdot(1 - \alpha)}\cdot k'\right\rceil > k$. Let $J' \subseteq J$ be the subset of dominant indices. Then, the $|J'|\times|J'|$ submatrix of $A$ induced by $J'$ is diagonally dominant, and thus full-rank by the Gershgorin circle theorem. This implies $\rank(A) \ge |J'| > k$, which contradicts our ealier conclusion that $\rank(A) \le k$.
\end{proof}

\subsection{Bounding the Number of Dominant Subsets}
Next, we upper bound $|\F_t^{(s)}|$ using Lemma~\ref{lemma:comb-constraint}. We use the following definition of packings.
\begin{definition}\label{def:packing}
    A $(d, s, \Delta)$-packing is a subset of $\binom{[d]}{s}$ such that the intersection of any two different sets has size at most $\Delta - 1$.
\end{definition}

With Definition~\ref{def:packing}, the ``in particular'' part of Lemma~\ref{lemma:comb-constraint} simply states that ``$\F_t^{(s)}$ does not contain a $(d, s, \Delta + 1)$-packing of size $k'$''.

Let $m(d, s, \Delta)$ denote the size of the maximum $(d, s, \Delta)$-packing. The following lemma gives an upper bound on $|\F_t^{(s)}|$ in terms of the packing size.

\begin{lemma}\label{lemma:size-upper-bound}
    For $k'$ and $\Delta$ chosen as in Lemma~\ref{lemma:comb-constraint},
    \[
        |\F_t^{(s)}|\le \frac{k' - 1}{m(d, s, \Delta + 1)} \cdot \binom{d}{s}.
    \]
\end{lemma}
Lemma~\ref{lemma:size-upper-bound} implies that whenever $m(d, s, \Delta + 1) \gg k'$, $\F_t^{(s)}$ only contains a negligible fraction of $\binom{[d]}{s}$. The lemma is proved by studying the intersection between $\F_t^{(s)}$ and a $(d, s, \Delta + 1)$-packing after a random permutation over the elements. This technique appears in a proof of the Erd\"os-Ko-Rado theorem~\cite{EKR61} due to Katona~\cite{Katona72}; also see \cite[Section 7.2]{Jukna11} for an exhibition.

\begin{proof}
    Let $M \coloneqq m(d, s, \Delta + 1)$ and $\{A_1, A_2, \ldots, A_M\}$ be a $(d, s, \Delta + 1)$-packing. Let $\pi$ be a permutation of $[d]$ chosen uniformly at random. Note that for each $i \in [M]$, $\pi(A_i) \coloneqq \{\pi(e): e \in A_i\}$ is still a size-$s$ subset of $[d]$. We examine the quantity
    \[
        \Ex{\pi}{\sum_{S \in \F_t^{(s)}}\sum_{i=1}^{M}\1{S = \pi(A_i)}}.
    \]
    On one hand, for each fixed $\pi$, $\pi(A_1), \pi(A_2) \ldots, \pi(A_M)$ is still a $(d, s, \Delta + 1)$-packing, and the double summation inside the expectation is simply the size of $\F_t^{(s)} \cap \{\pi(A_1), \pi(A_2) \ldots, \pi(A_M)\}$. Such an intersection must have size at most $k' - 1$ in light of Lemma~\ref{lemma:comb-constraint}.
    On the other hand, the quantity is equal to, by the linearity of expectation,
    \[
        \sum_{S \in \F_t^{(s)}}\sum_{i=1}^{M}\pr{\pi}{S = \pi(A_i)}
    =   \frac{|\F_t^{(s)}|\cdot M}{\binom{d}{s}}.
    \]
    Therefore, we have $|\F_t^{(s)}| \le \frac{k' - 1}{M}\cdot\binom{d}{s}$.
\end{proof}

Finally, the following lemma states two simple lower bounds on the packing number $m(d, s, \Delta)$.

\begin{lemma}\label{lemma:packing-number}
    Suppose that $q$ is a prime power, $1 \le \Delta \le s \le q$ and $d \ge sq$. We have
    \[
        m(d, s, \Delta) \ge q^\Delta.
    \]
    Furthermore, for any $d$ and $s$, we have $m(d, s, 1) \ge \lfloor d/s\rfloor$.
\end{lemma}

\begin{proof}
    Let $\bbF_q$ be the finite field of size $q$.
    For each $a = (a_0, a_1, \ldots, a_{\Delta - 1}) \in \bbF_q^\Delta$, we consider the polynomial $p_a(x) \coloneqq a_0 + a_1x + \cdots + a_{\Delta-1}x^{\Delta-1}$ over $\bbF_q$. Since $s \le q$, we can find $s$ distinct elements $x_1, x_2, \ldots, x_s \in \bbF_q$. We then define the set associated with $a$ as
    \[
        S_a \coloneqq \{(1, p_a(x_1)), (2, p_a(x_2)), \ldots, (s, p_a(x_s))\}.
    \]
    Clearly, $S_a$ is a size-$s$ subset of $[s]\times\bbF_q$. Furthermore, for $a \ne a'$, since $p_a(x) - p_{a'}(x)$ is a non-zero polynomial of degree $< \Delta$, $p_a(x)$ and $p_{a'}(x)$ agree on at most $\Delta - 1$ points. This implies $|S_a \cap S_{a'}| \le \Delta - 1$. Since $d \ge sq$, via an arbitrary injection from $[s] \times \bbF_q$ to $[d]$, $\{S_a: a\in\bbF_q^{\Delta}\}$ gives a $(d, s, \Delta)$-packing of size $q^\Delta$.

    Finally, the ``furthermore'' part follows from the $(d, s, 1)$-packing given by $\{1, 2, \ldots, s\}$, $\{s + 1, s + 2, \ldots, 2s\}$, $\ldots$, $\{s(\lfloor d / s\rfloor -1) + 1, \ldots, s\lfloor d / s\rfloor\}$.
\end{proof}

Combining Lemmas \ref{lemma:nec-cond}, \ref{lemma:comb-constraint}, \ref{lemma:size-upper-bound}~and~\ref{lemma:packing-number} immediately gives part of Theorem~\ref{thm:sparse-recovery}.
\begin{proof}[Proof of Theorem~\ref{thm:sparse-recovery} (the second and third bounds)]
    By Lemma~\ref{lemma:nec-cond} and the definition of $\F_t^{(s)}$, the probability $\pr{}{\|x^* - x\|_1 \ge t}$ is upper bounded by $|\F_t^{(s)}|/\binom{d}{s}$, which is, by Lemma~\ref{lemma:size-upper-bound}, at most $(k'-1)/m(d, s, \Delta + 1)$. Choosing $\alpha = 1/4$ in Lemma~\ref{lemma:comb-constraint} gives $k' = 12k + 1$, $\Delta = \lfloor t / (48k)\rfloor$ and
    \[
        \pr{}{\|x^* - x\|_1 \ge t}
    \le \frac{12k}{m(d, s, \lfloor t / (48k)\rfloor + 1)}.
    \]
    Now we apply Lemma~\ref{lemma:packing-number} by choosing $q$ as the unique power of $2$ in $\left(\frac{d}{2s}, \frac{d}{s}\right]$. Since $s \le \sqrt{d/2}$, we have $s \le \frac{d}{2s} < q$ and $d \ge sq$, i.e., the preconditions of Lemma~\ref{lemma:packing-number} are satisfied, and we get
    \[
        m(d, s, \lfloor t / (48k)\rfloor + 1)
    \ge q^{\lfloor t / (48k)\rfloor + 1}
    \ge \left(\frac{d}{2s}\right)^{\lfloor t / (48k)\rfloor + 1}.
    \]
    This gives the third bound $\pr{}{\|x^* - x\|_1 \ge t} \le 12k \cdot \left(\frac{2s}{d}\right)^{\lfloor t/(48k)\rfloor + 1}$.

    To prove the second bound, we take $t_0 = \min\{t, 24k\}$, and the above argument gives
    \[
        \pr{}{\|x^* - x\|_1 \ge t}
    \le \pr{}{\|x^* - x\|_1 \ge t_0}
    \le \frac{12k}{m(d, s, 1)}.
    \]
    By the ``furthermore'' part of Lemma~\ref{lemma:packing-number}, $m(d, s, 1) \ge \lfloor d / s\rfloor \ge d/(2s)$. This bounds the last term above by $24ks/d$.
\end{proof}

\subsection{Applying the Stricter Constraint}
To prove the first bound in Theorem~\ref{thm:sparse-recovery}, we need to apply the first part of Lemma~\ref{lemma:comb-constraint} (instead of the ``in particular'' part) to obtain the following alternative upper bound on $|\F_t^{(s)}|$.
\begin{lemma}\label{lemma:size-upper-bound-better}
    For $k'$ and $t'$ chosen as in Lemma~\ref{lemma:comb-constraint},
    \[
        |\F_t^{(s)}| \le \binom{s(k'-1)}{t'+1} \cdot \binom{d - t' - 1}{s - t' - 1}.
    \]
\end{lemma}

To prove Lemma~\ref{lemma:size-upper-bound-better}, we define the following notion of a perfect matching and state a characterization of its existence from Hall's theorem.

\begin{definition}[Perfect $s$-matching.]\label{def:perfect-matching}
    A collection of $n$ sets $S_1, S_2, \ldots, S_n$ can be perfectly $s$-matched if there are pairwise disjoint sets $T_1 \in \binom{S_1}{s}$, $T_2 \in \binom{S_2}{s}$, $\ldots$, $T_n \in \binom{S_n}{s}$.
\end{definition}

With this definition, the main part of Lemma~\ref{lemma:comb-constraint} can be re-written as ``no $k'$ sets in $\F_t^{(s)}$ can be perfectly $(s - t')$-matched''.

We can view the above definition in terms of a bipartite graph $(X, Y, E)$, in which the vertex sets are $X = [n]$ and $Y = \bigcup_{i \in [n]}S_i$, while $(x, y) \in X \times Y$ is present in the edge set $E$ if and only if $y \in S_x$. The $s$-matching in the above definition then corresponds to a subset of $E$ such that each $X$-vertex has exactly $s$ incident edges, while each $Y$-vertex has at most one incident edge.

The following immediate corollary of Hall's marriage theorem characterizes the existence of perfect $s$-matchings.

\begin{lemma}[Hall's marriage theorem~\cite{Hall35}]\label{lemma:halls-theorem}
    $S_1, S_2, \ldots, S_n$ can be perfectly $s$-matched if and only if for every $J \subseteq [n]$,
    \[
        \left|\bigcup_{i \in J}S_i\right| \ge |J|\cdot s.
    \]
\end{lemma}

\begin{proof}
    The ``only if'' direction is easy: If a perfect $s$-matching is witnessed by $T_1, T_2, \ldots, T_n$, it holds for every $J \subseteq [n]$ that
    \[
        \left|\bigcup_{i \in J}S_i\right| \ge \left|\bigcup_{i \in J}T_i\right| = \sum_{i \in J}|T_i| = |J| \cdot s.
    \]

    For the ``if'' direction, suppose that the inequality holds for every $J \subseteq [n]$. We construct another bipartite graph $G = (X, Y, E)$ by representing each of the $n$ sets with $s$ vertices. Formally, $X = [n] \times [s]$, $Y = \bigcup_{i = 1}^{n}S_i$ and
    \[
        E = \{((i, j), y) \in X \times Y: y \in S_i\}.
    \]
    Let $N(x) \coloneqq \{y \in Y: (x, y) \in E\}$ denote the neighbours of $x \in X$.
    
    We claim that $|\bigcup_{x \in J}N(x)| \ge |J|$ holds for every $J \subseteq X$. Assuming this claim, we are done: Hall's marriage theorem shows that $G$ has a matching of size $|X|$. In other words, there exists an injection $f: [n]\times [s] \to \bigcup_{i=1}^{n}S_i$ such that $f(i, j) \in S_i$. Then, picking $T_i \coloneqq \{f(i, j): j \in [s]\}$ shows that $S_1, S_2, \ldots, S_n$ can be perfectly $s$-matched.
    
    To verify the claim, let $J' \coloneqq \{i \in [n]: \exists j \in [s], (i, j) \in J\}$ be the projection of $J$ onto the first coordinate. Clearly, $|J| \le |J'| \cdot s$. By our construction of $G$,
    \[
        \bigcup_{x \in J}N(x)
    =   \bigcup_{(i, j) \in J}S_i
    =   \bigcup_{i \in J'}S_i.
    \]
    Then, our assumption implies that the cardinality of the last set above is at least $|J'| \cdot s \ge |J|$. This establishes the inequality $|\bigcup_{x \in J}N(x)| \ge |J|$ and completes the proof.
\end{proof}

Now we are ready to prove Lemma~\ref{lemma:size-upper-bound-better}.

\begin{proof}[Proof of Lemma~\ref{lemma:size-upper-bound-better}]
    Let $k^*$ be the largest integer such that there exist $S_1, S_2, \ldots, S_{k^*} \in \F_t^{(s)}$ that can be perfectly $(s - t')$-matched. By Lemma~\ref{lemma:comb-constraint}, $k^*$ is at most $k' - 1$. Let $U \coloneqq S_1 \cup S_2 \cup \cdots \cup S_{k^*}$ be the union of the $k^*$ sets. We have $|U| \le sk^* \le s(k' - 1)$.
    
    In the rest of the proof, we will show that every $T \in \F_t^{(s)}$ must intersect $U$ at $\ge t' + 1$ elements. The lemma would then immediately follow, since every member of $\F_t^{(s)}$ can be specified by choosing a size-$(t'+1)$ subset of $U$ and then choosing $s - t' - 1$ from the remaining $d - t' - 1$ elements in $[d]$. The number of ways of doing this is at most
    \[
        \binom{|U|}{t' + 1} \cdot \binom{d - t' - 1}{s - t' - 1}
    \le \binom{s(k'-1)}{t' + 1} \cdot \binom{d - t' - 1}{s - t' - 1}
    \]
    as desired.

    Suppose towards a contradiction that some $T \in \F_t^{(s)}$ satisfies $|T \cap U| \le t'$. We write $S_{k^* + 1}\coloneqq T$ and show that $S_1$ through $S_{k^*+1}$ can also be perfectly $(s - t')$-matching, which violates the maximality of $k^*$. In light of Lemma~\ref{lemma:halls-theorem}, it suffices to prove $\left|\bigcup_{i \in J}S_i\right| \ge (s-t')|J|$ for every $J \subseteq [k^* + 1]$. If $k^* + 1 \notin J$, the assumption that $S_1, \ldots, S_{k^*}$ can be perfectly $(s - t')$-matched together with Lemma~\ref{lemma:halls-theorem} gives $\left|\bigcup_{i \in J}S_i\right| \ge (s - t')|J|$. If $k^* + 1 \in J$, the same argument gives
    \[
        \left|\bigcup_{i \in J\setminus \{k^* + 1\}}S_i\right| \ge (s - t')(|J| - 1).
    \]
    Furthermore, we have
    \[
        \left|S_{k^*+1} \cap \left(\bigcup_{i \in J\setminus \{k^* + 1\}}S_i\right)\right|
    \le \left|S_{k^*+1} \cap \left(\bigcup_{i \in [k^*]}S_i\right)\right|
    =   |T \cap U| \le t'.
    \]
    The two inequalities above together give
    \[
        \left|\bigcup_{i \in J}S_i\right|
    =   \left|\bigcup_{i \in J\setminus \{k^* + 1\}}S_i\right| + |S_{k^* + 1}| - \left|S_{k^*+1} \cap \left(\bigcup_{i \in J\setminus \{k^* + 1\}}S_i\right)\right|
    \ge (s - t')(|J| - 1) + s - t'
    =   (s - t')|J|.
    \]
    Therefore, applying Lemma~\ref{lemma:halls-theorem} again shows that $S_1, \ldots, S_{k^*+1}$ can be perfectly $(s - t')$-matched, which leads to a contradiction.
\end{proof}

Finally, we prove the first bound in Theorem~\ref{thm:sparse-recovery}.
\begin{proof}[Proof of Theorem~\ref{thm:sparse-recovery} (the first bound)]
    Lemmas \ref{lemma:nec-cond}~and~\ref{lemma:size-upper-bound-better} together imply that for some $k'$ and $t'$ chosen as in Lemma~\ref{lemma:comb-constraint}:
    \begin{align*}
        \pr{}{\|x^* - x\|_1 \ge t}
    &\le\frac{|\F_t^{(s)}|}{\binom{d}{s}} \tag{Lemma~\ref{lemma:nec-cond}}\\
    &\le\frac{\binom{s(k'-1)}{t'+1}\cdot\binom{d - t' - 1}{s - t' - 1}}{\binom{d}{s}}\tag{Lemma~\ref{lemma:size-upper-bound-better}}\\
    &\le\frac{[s(k'-1)]^{t'+1}}{(t'+1)!}\cdot\frac{s}{d}\cdot\frac{s-1}{d-1}\cdots\frac{s-t'}{d-t'} \tag{$\binom{n}{m} \le n^{m}/m!$ and $\binom{n}{m} = \frac{n}{m}\binom{n-1}{m-1}$}\\
    &\le\frac{1}{(t'+1)!}\cdot\left[\frac{s^2(k'-1)}{d}\right]^{t'+1}.
    \end{align*}
    The bound then follow from choosing $\alpha = 1/4$, $k' = 12k + 1$ and $t' = \lfloor t/4\rfloor$ in Lemma~\ref{lemma:comb-constraint}.
\end{proof}

\section{Recovery of a Gaussian Dataset}
In this section, we prove our result for robust subspace recovery (Theorem~\ref{thm:subspace-recovery}) and then apply it to solve the full recovery problem (Problem~\ref{prob:mean-est}).

\subsection{Preliminaries}
Our subspace recovery algorithm builds on the previous results on robust mean estimation and robust linear regression, which we state below. The following simple lemma states that the sample median is a robust mean estimator for one-dimension Gaussians.
\begin{lemma}\label{lemma:1-d-mean-est}
    Suppose that $x_1, x_2, \ldots, x_n \in \R$ are independent samples from $\N(\mu, \sigma^2)$, and the corrupted copies $\tilde x_1, \ldots, \tilde x_n \in \R$ satisfy $|\{i \in [n]: \tilde x_i \ne x_i\}| \le \eps n$ for some $\eps < 1/2 - \sqrt{\ln(2/\delta)/(2n)}$. Let $m$ be the median of $\tilde x_1, \ldots, \tilde x_n$. Then, with probability $1 - \delta$,
    \[
        |m - \mu| \le \sigma\gamma,
    \]
    where $\gamma$ is the $(1/2 + \eps + \sqrt{\ln(2/\delta)/(2n)})$-quantile of $\N(0, 1)$. In particular, for sufficiently large $n$ and sufficiently small $\eps$, $|m - \mu| \le O(\sigma (\eps + \sqrt{\log(n)/n}))$ holds with probability $1 - 1/\poly(n)$.
\end{lemma}

\begin{proof}
    We upper bound the probability of $m > \mu + \sigma\gamma$ by $\delta/2$; the probability of $m < \mu - \sigma\gamma$ will also be at most $\delta / 2$ by symmetry. Note that $m > \mu + \sigma\gamma$ holds only if $\ge n/2$ values among $\tilde x_1, \ldots, \tilde x_n$ are above $\mu + \sigma\gamma$. Since at most $\eps n$ values are corrupted, this further implies that $\ge (1/2 - \eps)n$ numbers among $x_1, \ldots, x_n$ are larger than $\mu + \sigma\gamma$.

    For each $i \in [n]$, $\pr{}{x_i > \mu + \sigma\gamma} = \pr{X \sim \N(0, 1)}{X > \gamma} = \frac{1}{2} - \eps - \sqrt{\ln(2/\delta)/(2n)}$. Thus, by a Chernoff bound,
    \[
        \pr{}{m > \mu + \sigma\gamma}
    \le \pr{}{\frac{1}{n}\sum_{i=1}^{n}\1{x_i > \mu + \sigma\gamma} \ge \frac{1}{2} - \eps}
    \le \frac{\delta}{2}.
    \]

    Finally, the ``in particular'' part follows from the observation that $\eps < 1/2 - \sqrt{\ln(2/\delta)/(2n)}$ holds for all $\eps \le 1/4$, $\delta = 1/\poly(n)$ and sufficiently large $n$, together with the fact that the $(1/2 + \alpha)$-quantile of $\N(0, 1)$ is $O(\alpha)$ for all $\alpha$ bounded away from $1/2$.
\end{proof}

Another building block in our algorithm is from the literature on robust linear regression~\cite{KKM18,DKS19,DKKLSS19,PSBR20}. In particular, we will invoke the algorithm of~\cite{DKS19} to identify the linear dependence among several coordinates w.r.t.\ the hidden subspace $U$ (or verify that they are independent). The following lemma is an immediate consequence of the main result of~\cite{DKS19}. 
\begin{lemma}[Robust subspace recovery; implicit in~\cite{DKS19}]\label{lemma:cov-est}
    Let $\Sigma$ be a $d \times d$ positive semidefinite matrix such that the restriction of $\Sigma$ to the first $d-1$ coordinates is full-rank. Suppose that $x^{(1)}, \ldots, x^{(n)} \in \R^d$ are independent samples from $\N(0, \Sigma)$, and the corrupted copies $\tilde x^{(1)}, \ldots, \tilde x^{(n)} \in \R^d$ satisfy $|\{i \in [n]: \tilde x^{(i)} \ne x^{(i)}\}| \le \eps n$ for sufficiently small $\eps > 0$. There is an efficient algorithm that, when $n = \tilde\Omega(d^2/\eps^2)$,  correctly outputs $\rank(\Sigma)$ with high probability. Furthermore, when $\rank(\Sigma) = d - 1$, the algorithm outputs $c \in \R^{d-1}$ such that $\Sigma\begin{bmatrix}c\\-1 \end{bmatrix} = 0$.
\end{lemma}

\begin{proof}
    Let $\Sigma'$ be the $(d-1)\times(d-1)$ restriction of $\Sigma$ to the first $d - 1$ coordinates. Suppose that $\rank(\Sigma) = d - 1$ and $u$ is a non-zero vector in the kernel of $\Sigma$. Since $\Sigma'$ is full-rank by assumption, $u_d$ must be non-zero, so we may take $u = \begin{bmatrix}c\\-1\end{bmatrix}$ for some $c \in \R^{d-1}$ without loss of generality. Then, any clean data $x^{(i)}$ drawn from $\N(0, \Sigma)$ must satisfy $u^{\top}x^{(i)} = 0$. Equivalently, $x^{(i)}_d$ is exactly a linear function of the first $d-1$ coordinates given by weight vector $c$. Then, the main result of~\cite{DKS19} gives an algorithm that, as long as $n = \tilde\Omega(d^2/\eps^2)$, exactly recovers $c$ from the corrupted data efficiently. Furthermore, when $c$ is correctly recovered, we can verify that $\tilde x^{(i)}_d = c^{\top}\tilde x^{(i)}|_{[d-1]}$ holds for the majority of $i \in [d]$.

    On the other hand, when $\Sigma$ is full-rank, no matter what weight vector $c$ is returned by the robust linear regression algorithm, the equation $\tilde x^{(i)}_d = c^{\top}\tilde x^{(i)}|_{[d-1]}$ will not hold for the majority of the observed data. Thus, we can correctly decide whether $\Sigma$ is full-rank and also return the vector $c$ in the rank-deficient case.
\end{proof}

\subsection{Subspace Recovery}
Our subspace recovery algorithm is stated in Algorithm~\ref{algo:subspace-recovery}. Let $U$ be the column space of the unknown covariance $\Sigma$, and let $U' = \Span(U \cup \{\mu\})$. Algorithm~\ref{algo:subspace-recovery} maintains a set $J$ of coordinates that are independent w.r.t.\ $U$, and greedily add new indices into it. We check whether adding an index $i \in [d]$ leads to a linear dependence using the algorithm from Lemma~\ref{lemma:cov-est}. If so, Lemma~\ref{lemma:cov-est} enables us to identify a vector $c'$ that is orthogonal to the subspace $U$, in which case we add $c'$ to set $V$; otherwise, we include the independent coordinate $i$ in $J$. Ideally, we end up with $|J| = k$ and $|V| = d - k$, in which case we can identify the subspace $U$ using $[\Span(V)]^{\bot}$. At this point, we would be done if the mean vector $\mu$ is also in $U$; the case that $\mu \notin U$ is handled via a simple post-processing based on $V$.

\begin{algorithm2e}[h]
    \caption{Subspace Recovery under Coordinate-Level Corruption}
    \label{algo:subspace-recovery}
    \KwIn{Corrupted data $\tilde x^{(1)}, \tilde x^{(2)}, \ldots, \tilde x^{(n)}$ in the setup of Problem~\ref{prob:mean-est}.}
    \KwOut{An estimate $\hat U'$ of the subspace $U' = \Span(U \cup \{\mu\})$ in which the clean data lie.}
    $J \gets \emptyset$; $V \gets \emptyset$\;
    \For{$i \in [d]$} {
        \lFor{$j \in [n/2]$}{$y^{(j)} \gets (\tilde x^{(2j-1)} - \tilde x^{(2j)})|_{J \cup \{i\}}$}
        Run the algorithm from Lemma~\ref{lemma:cov-est} on $y^{(1)}, \ldots, y^{(n/2)}$\label{line:DKS}\;
        \uIf{algorithm reports full rank} {
            $J \gets J \cup \{i\}$\;
        } \Else {
            Let $c \in \R^{|J|}$ be the vector returned by the algorithm\;
            Set $c' \in \R^d$ such that $c'|_J = c$, $c'_i = -1$ and $c'|_{[d]\setminus(J\cup\{i\})} = 0$\;
            $V \gets V \cup \{c'\}$\;
        }
    }
    $\hat U \gets$ the orthogonal complement of $\Span(V)$\;
    \lFor{$v \in V$}{$y_v \gets$ the majority of $\{v^{\top}\tilde x^{(i)}: i \in [n]\}$}
    Let $\alpha \in \R^d$ be a solution to the linear system $\{v^{\top}\alpha = y_v: v \in V\}$\label{line:solve-linear}\;
    $\hat U' \gets \Span(\hat U \cup \{\alpha\})$\;
    \Return $\hat U'$\;
\end{algorithm2e}

The correctness of Algorithm~\ref{algo:subspace-recovery} follows from Lemma~\ref{lemma:cov-est} and some basic linear algebra. The proof proceeds by first proving that invoking the algorithm from Lemma~\ref{lemma:cov-est} allows us to correctly identify linearly dependent coordinates. Then, we argue that after the first for-loop, the sets $J$ and $V$ are constructed as expected, and thus $\hat U$ correctly recovers the subspace $U$ in which the Gaussian noises lie. Finally, we show that whenever $\mu \notin U$, the second part of the algorithm correctly augments the subspace to contain $\mu$.

Formally, we start by defining a ``good event'' $\Egood$ that implies the correctness of the algorithm. We will show that, over the randomness in both the clean data $x^{(1)}, \ldots, x^{(n)}$ and the corruption locations $S^{(1)}, \ldots, S^{(n)}$, $\Egood$ happens with high probability.

Write the covariance matrix $\Sigma$ as $A^{\top}A$ for some $k \times d$ matrix $A$. Let $a_i \in \R^k$ be the $i$-th column of $A$. We define sets $J_0, J_1, \ldots, J_d \subseteq [d]$ as:
\begin{itemize}
    \item $J_0 = \emptyset$.
    \item For every $i \in [d]$, if $a_i \in \Span(\{a_j: j \in S_{i-1}\})$, $J_i = J_{i-1}$; otherwise, $J_i = J_{i-1} \cup \{i\}$.
\end{itemize}
In other words, each $J_i$ corresponds to a basis among $\{a_1, a_2, \ldots, a_i\}$ obtained by adding vectors greedily. Clearly, the size of each $J_i$ is at most $k$. In our analysis below, $J_i$ will be the intended value of the set $J$ after the first $i$ iterations of the for-loop in Algorithm~\ref{algo:subspace-recovery}.

\begin{definition}
    $\Egood$ is defined as the event that the following happen simultaneously in an execution of Algorithm~\ref{algo:subspace-recovery}:
    \begin{itemize}
        \item For every $i \in [d]$,
        \[
            \sum_{j=1}^{n}\1{S^{(j)} \cap (J_{i-1} \cup \{i\}) \ne \emptyset} < \frac{n}{2}.
        \]
        In other words, the majority of the data are clean when restricted to the coordinates indexed by $J_{i-1} \cup \{i\}$.
        \item For every $i \in [d]$, the algorithm from Lemma~\ref{lemma:cov-est} gives the correct answer if $J = J_{i-1}$ holds on Line~\ref{line:DKS} of the $i$-th iteration of the for-loop.
    \end{itemize}
\end{definition}

We first show that $\Egood$ indeed happens with high probability, over the randomness in the clean data and the locations of the corruptions.

\begin{lemma}\label{lemma:good-event}
    $\pr{}{\Egood} \ge 1 - 1/d$.
\end{lemma}

\begin{proof}
    We first note that $J_0, J_1, \ldots, J_d$ are determined by the covariance $\Sigma$ alone, and thus independent of $x^{(1)}, \ldots, x^{(n)}$ and $S^{(1)}, \ldots, s^{(n)}$.

    We start by upper bounding the probability for the first condition to be violated at each $i \in [d]$. Note that for each $j \in [n]$, $\1{S^{(j)} \cup (J_{i-1}\cup\{i\}) \ne \emptyset}$ is an independent Bernoulli random variable with expectation at most $|J_{i-1}\cup\{i\}|\cdot (s/d) \le (k+1)s/d \le 2ks/d \le 2c_0$. Then, for sufficiently small $c_0$ and by a Chernoff bound, the probability for the inequality to be violated is at most $e^{-\Omega(n)}$. We can make this probability smaller than $1/(2d^2)$ by setting $n = \Omega(\log d)$. (This can be done while still satisfying $n = \tilde O(k^2)$, since the $\tilde O(\cdot)$ notation is allowed to hide $\polylog(d)$ factors.) By a union bound over $i \in [d]$, the total probability for the first condition to be violated is at most $1/(2d)$.

    For each $i \in [d]$, suppose that we run the algorithm from Lemma~\ref{lemma:cov-est} when $J = J_{i-1}$. Note that for each $j \in [n/2]$, $x^{(2j-1)} - x^{(2j)}$ follow the distribution $\N(0, 2\Sigma)$. Let $\Sigma'$ be the principal submatrix of $\Sigma$ indexed by $J_{i-1}\cup\{i\}$. Then, $y^{(j)} = (\tilde x^{(2j-1)} - \tilde x^{(2j)})|_{J_{i-1}\cup\{i\}}$ can be viewed as a corrupted copy of a sample from $\N(0, 2\Sigma')$, and the probability of corruption is at most $2\cdot|J_{i-1} \cup\{i\}|\cdot(s / d) \le 2(k+1)s/d \le 4ks/d \le 4c_0$. Therefore, for sufficiently small $c_0$ and by a Chernoff bound, at least a $2/3$ fraction of $y^{(1)}, \ldots, y^{(n/2)}$ are clean samples from $\N(0, 2\Sigma')$, except with an $e^{-\Omega(n)}$ probability. In this case, the algorithm from Lemma~\ref{lemma:cov-est}, with parameters $d$ and $\eps$ set to $k + 1$ and $1/3$, should give the correct answer with probability at least $1 - 1 / (4d^2)$. The probability for $\Egood$ not to happen because of the second condition is then at most
    \[
        d \cdot \left(e^{-\Omega(n)} + \frac{1}{4d^2}\right)
    \le d \cdot \left(\frac{1}{4d^2} + \frac{1}{4d^2}\right)
    =   \frac{1}{2d},
    \]
    where the first step holds for some appropriate $n = \Omega(\log d)$.

    Therefore, we conclude that $\pr{}{\Egood} \ge 1 - 1/(2d) - 1/(2d) = 1 - 1/d$.
\end{proof}

The following lemma states that conditioning on event $\Egood$, the set $J$ computed through the course of Algorithm~\ref{algo:subspace-recovery} matches the sets $J_0, J_1, \ldots, J_d$. Furthermore, the algorithm from Lemma~\ref{lemma:cov-est} allows us to identify the linear structure in $\Sigma$.

\begin{lemma}\label{lemma:correctness-of-J}
    Write the covariance matrix $\Sigma$ as $A^{\top}A$ for some $k \times d$ matrix $A$. Let $a_i \in \R^k$ be the $i$-th column of $A$. Conditioning on the event $\Egood$, the following holds for every $i \in [d]$:
    \begin{itemize}
        \item If $a_i \in \Span(\{a_j: j \in J_{i-1}\})$, the algorithm from Lemma~\ref{lemma:cov-est} outputs a vector $c \in \R^{|J_{i-1}|}$ such that $a_i = A|_{J_{i-1}} \cdot c$.
        \item After the $i$-th iteration of the for-loop, $J$ is equal to $J_i$.
    \end{itemize}
\end{lemma}

\begin{proof}
    We prove the lemma by induction. At $i = 0$, we clearly have $J = \emptyset = J_0$. Suppose that $i \ge 1$ and, after the first $i - 1$ iterations, $J$ is equal to $J_{i-1}$. The definition of $\Egood$ guarantees that the algorithm from Lemma~\ref{lemma:cov-est} returns the correct answer. In particular, if $a_i \in \Span(\{a_j: j \in J_{i-1}\})$, the algorithm outputs $c \in \R^{|J_{i-1}|}$ such that $\begin{bmatrix}c\\-1\end{bmatrix}$ is in the kernel of $\Sigma' = (A|_{J_{i-1} \cup \{i\}})^{\top}(A|_{J_{i-1} \cup \{i\}})$. This implies that $(A|_{J_{i-1} \cup \{i\}})\begin{bmatrix}c\\-1\end{bmatrix} = 0$, i.e., $a_i = A|_{J_{i-1}} \cdot c$. Furthermore, the algorithm reports full rank if and only if $a_i \notin \Span(\{a_j: j \in J_{i-1}\})$, which means that the algorithm correctly updates $J$ to $J_i$. This completes the inductive step.
\end{proof}

Recall that $U$ is the column space of $\Sigma$. The following lemma states the desired properties of the set $V$ after the for-loop.

\begin{lemma}\label{lemma:correctness-of-V}
    In Algorithm~\ref{algo:subspace-recovery}, conditioning on the event $\Egood$, the following holds after the for-loop over $i \in [d]$: $V$ contains $d - k$ linearly independent vectors, each of which is orthogonal to $U$ and supported over $J_{i-1} \cup \{i\}$ for some $i \in [d]$.
\end{lemma}

\begin{proof}
    By Lemma~\ref{lemma:correctness-of-J}, we have $|J| = \rank(A) = k$ after the for-loop. Since each iteration increases either $|J|$ or $|V|$ by $1$, $|V|$ is equal to $d - |J| = d - k$.

    Consider $i \in [d] \setminus J$ and the vector $c'$ added to $V$ in the $i$-th iteration. By Lemma~\ref{lemma:correctness-of-J} and our choice of $c'$ in Algorithm~\ref{algo:subspace-recovery}, the support of $c'$ is contained in $J_{i-1} \cup \{i\}$. By Lemma~\ref{lemma:correctness-of-J}, the vector $c'$ satisfies $Ac' = 0$ and thus, $\Sigma c' = A^{\top}Ac' = 0$, i.e., $c'$ is orthogonal to $U$.

    Finally, to prove that the vectors in $V$ are linearly independent, we note that the vector $c'$ added in the $i$-th iteration satisfies $c'_i = -1 \ne 0$ and $c'_j = 0$ for all $j > i$. This implies that these vectors are linearly independent: Consider the $(d - k)\times(d - k)$ matrix the columns of which are the vectors in $V$ restricted to indices in $[d] \setminus J$. Such a matrix is upper triangular and has non-zero entries on the diagonal, and is thus full-rank.
\end{proof}

\begin{proof}[Proof of Theorem~\ref{thm:subspace-recovery}]
Algorithm~\ref{algo:subspace-recovery} clearly runs in polynomial time, since it runs the algorithm from Lemma~\ref{lemma:cov-est} exactly $d$ times, solves a linear systems with $O(d)$ variables and equations, and performs a polynomial number of other elementary operations.

We prove that the algorithm is correct conditioning on $\Egood$. This is sufficient since Lemma~\ref{lemma:good-event} guarantees that $\Egood$ happens with high probability. By Lemma~\ref{lemma:correctness-of-V}, $\Span(V)$ is $(d-k)$-dimensional and $U$ is contained in $\hat U 
= [\Span(V)]^{\bot}$. Then, since $\dim(\hat U) = d - \dim(\Span(V)) = k = \dim(U)$, $\hat U$ is exactly equal to $U$.

Recall that each clean data point $x^{(i)}$ is the sum of $\mu$ and some vector in subspace $U$. Thus, as each $v \in V$ is orthogonal to $U$, $v^{\top}x^{(i)}$ gives $v^{\top}\mu$. Note that $v^{\top}\tilde x^{(i)} \ne v^{\top}x^{(i)}$ only if $S^{(i)}$ and the support of $v$ intersect. By Lemma~\ref{lemma:correctness-of-V}, the support of $v$ is contained in $J_{j-1} \cup \{j\}$ for some $j \in [d]$.
Event $\Egood$ guarantees that $\sum_{i=1}^{n}\1{S^{(i)} \cap (J_{j-1} \cup \{j\}) \ne \emptyset} < n/2$. Therefore, the majority of $v^{\top}\tilde x^{(i)}$ over $i \in [n]$ gives the correct value of $v^{\top}\mu$. Finally, since $U$ is exactly the orthogonal complement of $V$, every solution $\alpha$ to the linear system $\{v^{\top}\alpha = v^{\top}\mu: v \in V\}$ can be written as $\mu + u$ for some $u \in U$. Therefore, the subspace $\hat U' = \Span(\hat U \cup \{\alpha\})$ that Algorithm~\ref{algo:subspace-recovery} outputs is the desired subspace $U' = \Span(U \cup \{\mu\})$.
\end{proof}

\begin{remark}[Numerical stability]\label{remark:numerical-stability-detailed}
    As stated, Algorithm~\ref{algo:subspace-recovery} might be numerically unstable in the following steps: (1) Solving the linear system on Line~\ref{line:solve-linear}; (2) Running the robust linear regression algorithm of~\cite{DKS19} on Line~\ref{line:DKS}.

    The first step does not pose a problem. This is because our algorithm guarantees that every vector $v \in V$ is supported over $J \cup \{i\}$ for some $i \in [d] \setminus J$. Furthermore, $c_i = -1$. Thus, to satisfy the equation $v^{\top}\alpha = y_v$, we simply set $\alpha_j = 0$ for every $j \in J$ and $\alpha_i = -y_v$. This allows us to solve the linear equations in a numerically stable way.

    On the other hand, the stability of the first step relies on an appropriate condition number bound on $\Sigma$. Concretely, we assume that every full-rank principal sub-matrix of $\Sigma$ has eigenvalues bounded between $1/\kappa$ and $\kappa$. Then, when we apply the algorithm of~\cite{DKS19} via Lemma~\ref{lemma:cov-est}, the covariance $\Sigma'$ of the marginal distribution in the robust linear regression instance would be a full-rank sub-matrix of $\Sigma$. Note that the algorithm of~\cite{DKS19} starts with (robustly) estimating $\Sigma'$ and applying an appropriate linear transform so that the covariance becomes $\Theta(I)$. These two steps should not cause numerical issues since $\Sigma'$ is well-conditioned by our assumption.

    Furthermore, it is easy to bound the entries in the weight vector $c$ in the linear regression instance by $\poly(d, \kappa)$. Thus, verifying whether $x_i = c^{\top}x|_J$ holds for the clean data can also be done in a numerically stable way. In the positive case, $x_i - c^{\top}x|_J$ would be a Gaussian distribution with variance $1/\poly(\kappa)$, as long as our estimate for $c$ is $1/\poly(d, \kappa)$-accurate. In the negative case, since the restriction of $\Sigma$ to indices in $J \cup \{i\}$, denoted by $\Sigma''$, is well-conditioned by assumption, $x_i - c^{\top}x|_J$ follows a Gaussian with variance
    \[
        \begin{bmatrix}c^{\top}&-1\end{bmatrix}\Sigma''\begin{bmatrix}c\\-1\end{bmatrix}
    \ge \frac{1}{\kappa}\cdot(1 + \|c\|_2^2)
    = \Omega(1/\kappa).
    \]
    Therefore, if the sample $x$ is given with $1/\poly(d, \kappa)$ accuracy, these two cases can still be distinguished.

    Conversely, it is easy to see that the bit complexity must have a $\log\kappa$ dependence---for instance, to distinguish a degenerate covariance $\begin{bmatrix}1&1\\1&1\end{bmatrix}$ from a close-to-degenerate one $\begin{bmatrix}1&1\\1&1 + \kappa\end{bmatrix}$, we need to distinguish whether $x_1 - x_2$ is constantly zero or follows a Gaussian $\N(0, \kappa)$, which is impossible if $x_1$ and $x_2$ are given with $\ll \log\kappa$ bits of precision.

    In light of the informal discussion above, we conjecture that our proof of Theorem~\ref{thm:subspace-recovery} can be made robust to noises in the samples and floating point operations, as long as $\Theta(\log\kappa)$ bits of precision can be guaranteed, which seems to be a minimal assumption. We omit the formal investigation of this issue to future work, as it is orthogonal to the main focus of this paper.
\end{remark}

\subsection{Put Everything Together}
Now we prove Theorem~\ref{thm:mean-estimation} using Corollary~\ref{cor:expected-error} and Theorem~\ref{thm:subspace-recovery}. Formally, our algorithm is defined in Algorithm~\ref{algo:mean-est}: We first run Algorithm~\ref{algo:subspace-recovery} to recover the subspace $U'$ that contains all the clean data points. Then, we use the simple median algorithm from Lemma~\ref{lemma:1-d-mean-est} to obtain a crude estimate of $m_j \approx \mu_j$, which is then used to clip the observed data. Finally, we use the Basis Pursuit method to denoise the data points individually.

\begin{algorithm2e}[h]
    \caption{Gaussian Data Recovery under Coordinate-Level Corruption}
    \label{algo:mean-est}
    \KwIn{Corrupted data $\tilde x^{(1)}, \tilde x^{(2)}, \ldots, \tilde x^{(n)}$ in the setup of Problem~\ref{prob:mean-est}. Parameter $B_0 = \Theta(B\sqrt{\log(nd)})$.}
    \KwOut{Estimates $\hat x^{(1)}, \hat x^{(2)}, \ldots, \hat x^{(n)}$.}
    Run Algorithm~\ref{algo:subspace-recovery} to obtain subspace $U'$\;
    \For{$j \in [d]$} {
        $m_j \gets$ median of $\tilde x^{(1)}_j, \tilde x^{(2)}_j, \ldots, \tilde x^{(n)}_j$\;
        Truncate $\tilde x^{(1)}_j, \tilde x^{(2)}_j, \ldots, \tilde x^{(n)}_j$ to $[m_j - B_0, m_j + B_0]$\;
    }
    \For{$i \in [n]$} {
        Let $\hat x^{(i)}$ be an optimum of the following BP program:
        \begin{align*}
            \text{minimize}\quad&\|\hat x - \tilde x^{(i)}\|_1\\
            \text{subject to}\quad& \hat x \in U'
        \end{align*}
    }
    \Return $\hat x^{(1)}, \ldots, \hat x^{(n)}$\;
\end{algorithm2e}

\begin{proof}[Proof of Theorem~\ref{thm:mean-estimation}]
    For every $j \in [d]$, the clean entries $x^{(1)}_j, x^{(2)}_j, \ldots, x^{(n)}_j$ independently follow the Gaussian distribution with mean $\mu_j$ and variance $\Sigma_{jj} \le B^2$. Furthermore, each entry gets corrupted independently with probability $s/d$. Thus, by Lemma~\ref{lemma:1-d-mean-est}, the median of the observed entries, $m_j$, satisfies $|m_j - \mu_j| = O(B\cdot(s/d + \sqrt{\log(n)/n}))$ with high probability. Furthermore, with high probability, every clean entry $x^{(i)}_j$ is $O(B\sqrt{\log(nd)})$-close to $\mu_j$, and thus $|x^{(i)}_j - m_j| \le O(B\sqrt{\log(nd)}) + |m_j - \mu_j| \le B_0$ (for appropriately chosen $B_0 = \Theta(B\sqrt{\log(nd)})$). Therefore, truncating the $j$-th coordinate of every data point to the range $[m_j - B_0, m_j + B_0]$ does not alter the clean data entries. This reduces the problem to the case that the corruptions are bounded in magnitude by $B_0 + |m_j - \mu_j| = O(B\sqrt{\log(nd)})$.

    By Theorem~\ref{thm:subspace-recovery}, running Algorithm~\ref{algo:subspace-recovery} on the corrupted data gives the subspace $U' = \Span(U \cup \{\mu\})$ with high probability. Then, by Corollary~\ref{cor:expected-error}, running the BP method with subspace $U'$ gives estimates $\hat x^{(1)}$ through $\hat x^{(n)}$ such that $\Ex{}{\|\hat x^{(i)} - x^{(i)}\|_1}$ is bounded by $\tilde O(Bks/d\cdot\max\{ks^2/d, 1\})$.
\end{proof}

\section{Discussion and Open Problems}

In this section, we highlight a concrete open problem, the solutions of which could further strengthen the recovery guarantees presented in this work. We also discuss several directions of further exploration on robust algorithms under coordinate-level corruptions.

\subsection{Optimal Error under Weaker Assumptions?}
Recall that in Theorem~\ref{thm:mean-estimation}, the nearly-optimal error bound of $\tilde O(Bks/d)$ only holds in the $ks^2 = O(d)$ regime. In comparison, the rest of the proof would go through under the weaker assumption of $ks = O(d)$. We discuss a concrete approach to closing this gap.

This extra factor of $s$ can be traced back to Lemma~\ref{lemma:size-upper-bound-better}, which we restate below (after replacing $k'$ and $t'$ with $k$ and $t$ for brevity). Recall the definition of a perfect $s$-matching from Definition~\ref{def:perfect-matching}.

\vspace{6pt}
\noindent\textbf{Lemma~\ref{lemma:size-upper-bound-better}} (rephrased).\ {\it
    If $\F \subseteq \binom{[d]}{s}$ contains no $k$ sets that can be perfectly $(s-t)$-matched,
    \[
        |\F| \le \binom{s(k-1)}{t+1}\cdot\binom{d-t-1}{s-t-1}.
    \]
}
\vspace{6pt}

We conjecture that maximum possible size of $\F$ is achieved by one of a few simple constructions.

\begin{conjecture}\label{conj:size-bound}
    If $\F \subseteq \binom{[d]}{s}$ contains no $k$ sets that can be perfectly $(s-t)$-matched,
    \[
        |\F| \le \max_{i \in [s - t]}|\F_i|,
    \]
    where $\F_i \coloneqq \left\{S \in \binom{[d]}{s}: |S \cap [ik - 1]| \ge t + i\right\}$.
\end{conjecture}

We first note that each $\F_i$ defined in Conjecture~\ref{conj:size-bound} is a valid choice of $\F$---no $k$ sets $S_1, \ldots, S_k \in \F_i$ can be perfectly $(s - t)$-matched. This is because for any $T_j \subseteq S_j$ of size $s - t$, $|T_j \cap [ik - 1]|$ is at least $|S_j \cap [ik - 1]| - t \ge i$. By a pigeonhole argument, $T_1, T_2, \ldots, T_k$ cannot be disjoint. Thus, if $\max|\F_i|$ is indeed a valid upper bound, it cannot be improved any further. The idea behind the construction of $\F_i$ is to restrict a certain number of degrees of freedom (namely, $t + i$) to a small subset (namely, $[ik - 1]$), and then allow the remaining elements to be freely chosen. We remark that constructions that resemble the above have been shown or conjectured to be optimal for many natural problems in extremal combinatorics, including those in the Erd\H{o}s-Ko-Rado theorem and the Erd\H{o}s Matching Conjecture. We have also verified Conjecture~\ref{conj:size-bound} for $d \le 8$, $s \in [d]$, $k \in \{2, 3, \ldots, d\}$ and $t \in \{0, 1, \ldots, s - 1\}$, either manually or via an exhaustive search.

The upper bound in Conjecture~\ref{conj:size-bound} can be relaxed to the following closed-form expression:
\[
    |\F| \le \left(\frac{eks}{d}\right)^{t+1}\cdot\binom{d}{s}.
\]
Then, repeating the proof of Theorem~\ref{thm:sparse-recovery} and Corollary~\ref{cor:expected-error} would give an $\tilde O(ks/d)$ bound whenever $ks \le c_0d$ for some sufficiently small constant $c_0 > 0$. Details can be found in Appendix~\ref{sec:conjecture-corollary}.

\subsection{The $ks \gg d$ Regime?}
Our results for recovering the subspace and the individual data points crucially rely on the assumption that $ks = O(d)$. For subspace recovery, this assumption guarantees that when restricting the problem to any $O(k)$ coordinates, a majority of data points will be entirely clean, which allows the known robust estimators (under sample-level corruption) to achieve an accurate recovery.

Even when the subspace is given, when $ks \gg d$, the expected error bound in Corollary~\ref{cor:expected-error} reduces to $\tilde O(Bs \cdot (ks/d)^2)$, which is looser than the trivial upper bound of $O(Bs)$.\footnote{The optimum $x^*$ of BP must satisfy $\|x^* - \tilde x\|_1 \le \|x - \tilde x\|_1$, so $\|x^* - x\|_1 \le 2\|x - \tilde x\|_1 \le 2Bs$.} This seems to be an inherent limitation of the current approach. When $ks \gg d$, our constraint on $\F_t^{(s)}$ in Lemma~\ref{lemma:comb-constraint} is vacuous even when $t = s$: Even if $\F_t^{(s)} = \binom{[d]}{s}$, we still cannot find $k' = \Theta(k)$ sets that can be perfectly $(s - t')$-matched for $t' \le t/2 = s/2$, which is easily seen by the pigeonhole principle and that $k'(s-t') = \Omega(ks) \gg d$. Consequently, the current approach would fail to control the probability that BP incurs an $\Omega(s)$ error. On the other hand, natural choices for the subspace $U$ do not seem to give counterexamples in which $|\F_t^{(s)}|$ is close to $\binom{d}{s}$, even if $ks \gg d$. Therefore, obtaining non-trivial guarantees in the $ks \gg d$ regime \emph{might} still be possible based on the generalized version of RNP (Lemma~\ref{lemma:nec-cond}), though controlling the size of $\F_t^{(s)}$ will likely involve a quite different approach than the one based on diagonal dominance (in Lemma~\ref{lemma:comb-constraint}) explored in the present work.

\subsection{A Unified and More Natural Approach?}
Our solution to Problem~\ref{prob:mean-est} has two separate steps: first recovering the subspace that contains all the clean data, and then denoising each single data point with Basis Pursuit using this subspace information. Note that the distributional assumptions (i.e., Gaussianity) on the data are not used in the second step. It is then interesting to ask whether there is a more natural single-phase approach that achieves a similar performance guarantee.

For example, if we already know the low-rank covariance matrix $\Sigma$ (and thus the subspace $U$), and the goal is to estimate the mean $\mu$, the natural convex relaxation would be to set variables $\hat\mu \in \R^d$ and $y^{(1)}, y^{(2)}, \ldots, y^{(n)} \in U$, and minimize an appropriate combination of the following:
\begin{itemize}
    \item $\sum_{i=1}^{n}\|\tilde x^{(i)} - (\mu + y^{(i)})\|_1$, which is the $\ell_1$ surrogate for the sparsity.
    \item $\sum_{i=1}^{n}[y^{(i)}]^{\top}\Sigma^{-1}y^{(i)}$, which penalizes the noises $y^{(i)}$ that are unlikely to occur.
\end{itemize}

Taking one more step back, without directly modeling the Gaussianity of the data, does the ``vanilla'' robust PCA algorithm succeed for the setups in this paper? The near optimality of BP proved in our work might bring some optimism that the natural convex relaxation approach guarantees not only exact recovery under structural assumptions but also reasonable error bounds under the randomized location model.

\bibliographystyle{alpha}
\bibliography{main}

\appendix

\section{An Expected Error Bound}\label{sec:expcted-error}
\begin{proof}[Proof of Corollary~\ref{cor:expected-error}]
    We first reduce to the $B = 1$ case. If $x^*$ is an optimum of the BP program~\eqref{eq:basis-pursuit}, $x^* / B$ would be an optimum of $\|\hat x - \tilde x / B\|_1$ subject to $\hat x \in U$, since the two programs are equivalent. The latter is exactly the BP program for recovering $x / B \in U$ from the corrupted version $\tilde x / B$ with corruptions bounded by $1$. Thus, any upper bound on the expected error in the $B = 1$ case (namely, $\Ex{}{\|x^* / B - x / B\|_1}$) would give an upper bound on $\Ex{}{\|x^* - x\|_1}$ after multiplication by $B$.

    Fix $B = 1$, the subspace $U$ and the clean data point $x$. Let random variable $X$ denote the $\ell_1$ error incurred by the BP method, over the randomness in the support of the corruption. Let $t_0 \ge 0$ be an integer to be chosen later. Since $X$ is nonnegative, we have
    \begin{align*}
        \Ex{}{X}
    &=  \int_{0}^{4t_0}\pr{}{X \ge t}\rmd t + \int_{4t_0}^{+\infty}\pr{}{X \ge t}\rmd t\\
    &\le 4t_0\cdot\pr{}{X > 0} + \sum_{t=4t_0}^{+\infty}\pr{}{X \ge t}\\
    &\le 4t_0\cdot\frac{24ks}{d} + \sum_{t=4t_0}^{+\infty}\frac{1}{(\lfloor t/4\rfloor + 1)!}\cdot\left(\frac{12s^2k}{d}\right)^{\lfloor t/4\rfloor + 1} \tag{Theorem~\ref{thm:sparse-recovery}}\\
    &=  t_0\cdot\frac{96ks}{d} + 4\sum_{t=t_0 + 1}^{+\infty}\frac{1}{t!}\cdot\left(\frac{12s^2k}{d}\right)^{t}\\
    &\le t_0\cdot\frac{96ks}{d} + 4\sum_{t=t_0 + 1}^{+\infty}\left(\frac{12es^2k}{dt}\right)^{t},
    \end{align*}
    where the last step applies Stirling's approximation $n! > \sqrt{2\pi n}(n/e)^n > (n/e)^n$.
    
    We pick $t_0 = O(ks^2/d + \log d) = \tilde O(\max\{ks^2/d, 1\})$ such that $\frac{12eks^2}{dt_0} \le 1/2$ and $2^{-t_0} \le 1/d$. We get
    \[
        \Ex{}{X}
    \le t_0\cdot\frac{96ks}{d} + 4\sum_{t=t_0+1}^{+\infty}2^{-t}
    =   t_0 \cdot O\left(\frac{ks}{d}\right) + O\left(\frac{1}{d}\right)
    =   \tilde O\left(\frac{ks}{d}\cdot\max\left\{\frac{ks^2}{d}, 1\right\}\right).
    \]

    Finally, as discussed earlier, multiplying the bounds by $B$ gives the error bounds for the general $B$ case.
\end{proof}

\begin{remark}
While the proof above only uses the first two bounds in Theorem~\ref{thm:sparse-recovery}, the third bound is not implied by the first two. For instance, when $s = \sqrt{d/2}$ and $t = 48k$, the three bounds in Theorem~\ref{thm:sparse-recovery} get simplified to
\[
    \frac{(6k)^{12k+1}}{(12k+1)!},\quad\quad \frac{12\sqrt{2} \cdot k}{\sqrt{d}}, \quad\quad\frac{24k}{d},
\]
respectively. Then, for sufficiently large $d$ and $k = o(\log d)$, the last bound is strictly sharper than the first two.
\end{remark}

\section{Consequences of Conjecture~\ref{conj:size-bound}}\label{sec:conjecture-corollary}
We show that Conjecture~\ref{conj:size-bound} implies that the error bounds in Corollary~\ref{cor:expected-error} and Theorem~\ref{thm:mean-estimation} can be improved to $\tilde O(Bks/d)$ under the weaker assumption of $ks = O(d)$ (rather than $ks^2 = O(d)$).

Recall that $\F_i$ is defined as $\left\{S \in \binom{[d]}{s}: |S \cap [ik-1]| \ge t + i\right\}$. Thus, every set in $\F_i$ can be written as $A \cup B$, where $A$ is a size-$(t + i)$ subset of $[ik - 1]$ and $B$ is a size-$(s - t - i)$ subset of $[d]\setminus A$. This implies
\[
    |\F_i|
\le \binom{ik-1}{t+i}\binom{d-t-i}{s-t-i},
\]
and it follows that
\[
    \frac{|\F_i|}{\binom{d}{s}}
\le \binom{ik}{t+i}\cdot\frac{s}{d}\cdot\frac{s-1}{d-1}\cdots\frac{s-t-i+1}{d-t-i+1}
\le \frac{(ik)^{t+i}}{(t+i)!}\cdot\left(\frac{s}{d}\right)^{t+i}.
\]
Applying Stirling's approximation $n! > \sqrt{2\pi n}(n/e)^n > (n/e)^n$ and $1 - x \le e^{-x}$ gives
\[
    \frac{|\F_i|}{\binom{d}{s}}
\le \frac{i^{t+i}}{(t+i)^{t+i}}\cdot\left(\frac{eks}{d}\right)^{t+i}
=   \left(1 - \frac{t}{t+i}\right)^{t+i}\cdot\left(\frac{eks}{d}\right)^{t+i}
\le e^{-t}\left(\frac{eks}{d}\right)^{t+i}.
\]
Assuming $ks \le e^{-1}\cdot d$, we get
\[
    \frac{|\F_i|}{\binom{d}{s}}
\le e^{-t}\left(\frac{eks}{d}\right)^{t+1}
=   e \cdot \left(\frac{ks}{d}\right)^{t+1}.
\]

Then, we repeat the proof of  the first bound in Theorem~\ref{thm:sparse-recovery}. Recall that $\F_t^{(s)}$ denotes the set
\[
    \left\{S \in \binom{[d]}{s}: \exists u \in U\text{ such that }\|u\|_1 = 1, \sum_{i \in S}|u_i|\cdot\1{|u_i|\le 1/t} \ge \frac{1}{2}\right\}.
\]
By Lemma~\ref{lemma:nec-cond}, $\pr{}{\|x^* - x\|_1 \ge t}$ is at most $|\F_t^{(s)}|/\binom{d}{s}$. Then, Lemma~\ref{lemma:comb-constraint} implies that for $k' = 12k+1$ and $t' = \lfloor t/4\rfloor$, $\F_t^{(s)}$ contains no $k'$ sets that can be perfectly $(s-t)$-matched. Then, using Conjecture~\ref{conj:size-bound} and our calculation above,
\[
    \pr{}{\|x^* - x\|_1 \ge t}
\le \max_{i \in [s - t']}\frac{|\F_i|}{\binom{d}{s}}
\le e\cdot\left(\frac{k's}{d}\right)^{t'+1}
=   e\cdot\left[\frac{(12k+1)\cdot s}{d}\right]^{\lfloor t/4\rfloor + 1}.
\]
Finally, repeating the calculation from Appendix~\ref{sec:expcted-error} shows that when $B = 1$,
\begin{align*}
    \Ex{}{\|x^* - x\|_1}
&=  \int_0^{+\infty}\pr{}{\|x^* - x\|_1 \ge t}~\rmd t\\
&\le\sum_{t=0}^{+\infty}\pr{}{\|x^* - x\|_1 \ge t}\\
&\le e\cdot\sum_{t=0}^{+\infty}\left[\frac{(12k+1)\cdot s}{d}\right]^{\lfloor t/4\rfloor + 1}
=   O\left(\frac{ks}{d}\right),
\end{align*}
where the last step holds as long as $\frac{(12k+1)\cdot s}{d} \le 1 - c$ for any constant $c > 0$, which is implied by $ks \le c_0 \cdot d$ for sufficiently small $c_0 > 0$. This implies an expected error of $O(Bks/d)$ for Problem~\ref{prob:sparse}. Repeating the proof of Theorem~\ref{thm:mean-estimation} gives an $\tilde O(Bks/d)$ error in expectation.

\end{document}